\documentclass[letterpaper,11pt]{article}

\usepackage{latexsym}
\usepackage{graphicx}
\usepackage{xcolor}
\usepackage{float}
\usepackage{enumitem}
\usepackage{url}
\usepackage{bm}
\usepackage{tabularx}
\usepackage{array}
\usepackage[export]{adjustbox}
\usepackage[colorlinks=true, allcolors=blue]{hyperref}
\usepackage[letterpaper,top=2cm,bottom=2cm,left=2.5cm,right=2.5cm,marginparwidth=1.75cm]{geometry}
\usepackage{enumitem}

\usepackage{verbatim}

\usepackage[vlined]{algorithm2e}
\DontPrintSemicolon
\SetVlineSkip{1pt}
\NoCaptionOfAlgo
\SetKwIF{If}{ElseIf}{Else}{if}{:}{else if}{else:}{}%
\SetKwFor{For}{for}{:}{}%
\SetKwFor{While}{while}{:}{}%
\SetKwProg{Fn}{}{:}{}
\SetCommentSty{textrm}
\usepackage{amsmath}
\usepackage{amssymb}
\usepackage{amsfonts}
\usepackage{amsthm}

\newtheorem{theorem}{Theorem}[section]
\newtheorem{lemma}[theorem]{Lemma}
\newtheorem{claim}[theorem]{Claim}

\newtheorem{corollary}[theorem]{Corollary}
\newtheorem{definition}[theorem]{Definition}

\newcommand{\ignore}[1]{}
\newcommand{\remove}[1]{}

\newcommand{\sm}{\setminus}

\newcommand{\EQ}{\;=\;}
\newcommand{\LT}{\;<\;}
\newcommand{\GT}{\;>\;}
\newcommand{\LE}{\;\le\;}
\newcommand{\GE}{\;\ge\;}

\newcommand{\aaa}{\bm{a}}
\newcommand{\bbb}{\bm{b}}

\newcommand{\Array}{\mbox{\it Array}}
\newcommand{\Length}{\mbox{\it Length}}
\newcommand{\Get}{\mbox{\it Get}}
\newcommand{\Set}{\mbox{\it Set}}
\newcommand{\Grow}{\mbox{\it Grow}}
\newcommand{\Shrink}{\mbox{\it Shrink}}

\newcommand{\Allocate}{\mbox{\it Allocate}}
\newcommand{\Combine}{\mbox{\it Combine-Blocks}}
\newcommand{\Deallocate}{\mbox{\it Deallocate}}
\newcommand{\Copy}{\mbox{\it Copy}}
\newcommand{\Split}{\mbox{\it Split-Blocks}}
\newcommand{\Rebuild}{\mbox{\it Rebuild}}

\newcommand{\AAA}[2]{A[#1][#2]}

\newcommand{\NN}{\mathbb{N}}

\newcommand{\e}{{\rm e}}

\begin{document}

\newcommand\relatedversion{}

\title{Optimal resizable arrays \thanks{A preliminary version of this paper appeared in SOSA 2023.}}

\author{Robert E. Tarjan \thanks{Department of Computer Science, Princeton University, NJ, USA. E-mail: {\tt ret@princeton.edu}.  Research at Princeton University partially supported by an innovation research grant from Princeton and a gift from Microsoft.} \and Uri Zwick \thanks{Blavatnik School of Computer Science, Tel Aviv University, Tel Aviv, Israel. E-mail: {\tt zwick@tau.ac.il}. Research supported by ISF grant no.\ 2854/20.}}
\date{}

\maketitle

\begin{abstract}%
\setlength{\parindent}{0pt}
\setlength{\parskip}{3pt plus 2pt}\noindent
A \emph{resizable array} is an array that can \emph{grow} and \emph{shrink} by the addition or removal of items from its end, or both its ends, while still supporting constant-time \emph{access} to each item stored in the array given its \emph{index}. Since the size of an array, i.e., the number of items in it, varies over time,  space-efficient maintenance of a resizable array requires dynamic memory management. A standard doubling technique allows the maintenance of an array of size~$N$ using only $O(N)$ space, with $O(1)$ amortized time, or even $O(1)$ worst-case time, per operation. Sitarski and Brodnik et al.\ describe much better solutions that maintain a resizable array of size~$N$ using only $N+O(\sqrt{N})$ space, still with $O(1)$ time per operation. Brodnik et al.\ give a simple proof that this is best possible. 

We distinguish between the space needed for \emph{storing} a resizable array, and accessing its items, and the \emph{temporary} space that may be needed while growing or shrinking the array. For every integer $r\ge 2$, we show that $N+O(N^{1/r})$ space is sufficient for storing and accessing an array of size~$N$, if $N+O(N^{1-1/r})$ space can be used briefly during grow and shrink operations. Accessing an item by index takes $O(1)$ worst-case time while grow and shrink operations take $O(r)$ amortized time. Using an exact analysis of a \emph{growth game}, we show that for any data structure from a wide class of data structures that uses only $N+O(N^{1/r})$ space to store the array, the amortized cost of grow is $\Omega(r)$, even if only grow and access operations are allowed.
The time for grow and shrink operations cannot be made worst-case, unless $r=2$.
\end{abstract}

\section{Introduction}\label{S-intro}

Arrays and resizable arrays are perhaps the most widely used data structures. Surprisingly, we show that apparently not everything was already said about them. We describe simple new implementations of resizable arrays that in many cases are more memory-efficient than all implementations proposed so far. Our perspective is theoretical, but our proposed implementations may have some practical implications.

An \emph{array} of size~$N$ is a data structure that holds a sequence of items $a_0,a_1,\ldots,a_{N-1}$. The $i$-th item in the sequence, for any $0\le i<n$, can be retrieved or modified in constant time. An array is naturally implemented as a contiguous block of~$N$ words of memory, where each word holds an item, or a pointer to an item. Items are retrieved or modified in constant time using the \emph{random access} capabilities of the machine used. Arrays as defined here are inherently of a fixed size.

A \emph{resizable array}, also known as a \emph{dynamic array} or \emph{dynamic table}, is an array whose size can increase or decrease, while still allowing the retrieval or the modification of the $i$-th item for any $i$ in constant time. Decreasing the size of an array is easy. We simply do not use some of the memory allocated to the array. This, however, is not memory-efficient. Increasing the size of an array is harder, since the word just beyond the contiguous block allocated to the array may be in use for a different purpose. Thus, maintaining a resizable array requires dynamic memory management.

For simplicity, we assume throughout most of the paper that the resizable arrays grow and shrink only at their `far end'. Thus, when the size of the array is increased from~$N$ to~$N+1$ a new item~$a_N$ is added to the array. Similarly, when the size is decreased from~$N$ to~$N-1$, the item $a_{N-1}$ disappears. Such  resizable arrays form a generalization of a \emph{stack}. All our results extend to resizable arrays that can grow and shrink at both ends, however, and thus they give a generalization of a \emph{double-ended queue (deque)}. This follows easily: one can easily implement a double-ended resizable array using two single-ended resizable arrays.

The standard textbook solution for resizable arrays (see, e.g., Cormen et al.\ \cite[Section 17.4]{CLRS09}) is the \emph{doubling}, or \emph{geometric expansion/shrinking} technique. Start by allocating an array of some initial size. If the array is full, allocate a new array of, say, twice the size and copy all the items into it. When the array is, say, less than a quarter full, allocate a new array of, say, half the size and again copy all items into it. In both cases, the memory used by the old array is released and can be used for other purposes. A simple amortization argument shows that the \emph{amortized} cost of each size increase or decrease is $O(1)$. Each item can still be accessed and modified in $O(1)$ worst-case time. Furthermore, the time for an increase or  decrease operation can be made $O(1)$ \emph{worst-case} by a standard background-rebuilding technique.

According to Wikipedia \cite{enwiki:1078194534}, the geometric expansion/shrinking technique, with growth factors ranging from~$1.25$ to~$2$, is used to implement Java's ArrayList, Python's PyListObject, C++ STL's Vector and more.

The doubling technique allows storing a resizable array of size~$N$ using $O(N)$ space. This is nice, and may be sufficient in many applications. But it also means that at various times as much as 50\% of the memory allocated for an array is not used, even if only grow operations are considered. More generally, if a growth factor of $1+\alpha$ is used, a $\frac{1}{1+\alpha}$ fraction of memory may be wasted. By choosing a small value of $\alpha$ the fraction of wasted storage can be made arbitrarily small, but the amortized cost for a size increase becomes $O(\frac{1+\alpha}{\alpha})$, i.e., larger and larger.

Is it possible to store a resizable array of size~$N$ using $N+o(N)$ space while still maintaining $O(1)$ time per operation? 
Sitarski \cite{Sitarski96} and Brodnik et al.\ \cite{BCDMS99} answered this question affirmatively by describing simple implementations that use only $N+O(\sqrt{N})$ space. (See more details in Section~\ref{S-previous}.) Brodnik et al.\ \cite{BCDMS99} also show that in a certain sense this is optimal. More specifically, they show that any resizable array implementation must \emph{sometime} use $N+\Omega(\sqrt{N})$ space. (We review their simple argument in Section~\ref{S-lower}.)

We propose distinguishing between the amount of space needed to \emph{store} a resizable array currently of size~$N$, supporting retrieving and modifying items in $O(1)$ worst-case time, and the space that may be needed \emph{temporarily} to increase or decrease the size of the array, in $O(1)$ amortized time. In many settings it is preferable to have a more compact representation of the array, even if more storage is required, temporarily, for resizing the array. For example, an application may use many resizable arrays, only one of which is resized at any given moment. 

Our first implementation to `break' the lower bound of Brodnik et al.\ \cite{BCDMS99} is a simple implementation that uses only $N+O(N^{1/3})$ space to store an array, while temporarily using $N+O(N^{2/3})$ space to resize it. Items can be accessed in $O(1)$ worst-case time, and increasing or decreasing the size of the array takes $O(1)$ amortized time. 
Furthermore, we show that this is essentially optimal. Any implementation that uses only $N+O(N^{1/3})$ space for storing an array must, at certain times, use $N+\Omega(N^{2/3})$ space during resizing. 

More generally we show that, for every $r\ge 2$, it is possible to store a resizable array using $N+O(N^{1/r})$ space, while needing only $N+O(N^{1-1/r})$ space to resize it. Retrieving and modifying items by index takes $O(1)$ worst-case time, while resizing takes $O(r)$ amortized time. Since background rebuilding cannot be used, the $O(r)$ amortized time for resizing, which we show is optimal, cannot be made worst-case, unless $r=2$, in which case the implementation becomes very similar to the implementations of Sitarski \cite{Sitarski96} and Brodnik et al.\ \cite{BCDMS99}.

The rest of the paper is organized as follows. In Section~\ref{S-resize} we give a more precise definition of the problem and the computational model used. In Section~\ref{S-previous} we describe previous work, mostly that of Sitarski \cite{Sitarski96} and Brodnik et al.\ \cite{BCDMS99}. In Section~\ref{S-lower} we describe the simple lower bound of Brodnik et al.\ \cite{BCDMS99} and our extension of it. In Section~\ref{S-simple} we give our simple $r=3$ solution, i.e., $N+O(N^{1/3})$ space for storing, $N+O(N^{2/3})$ space for resizing. In Section~\ref{S-general-r} we describe, for any $r\ge 2$, an implementation that uses $N+O(rN^{1/r})$ space to store the array, while using only $N+O(N^{1-1/r})$ space during resize operations. The amortized cost of grow and shrink operations is $O(r)$. In Section~\ref{S-transform} we give a simple transformation that can be used to tune the implementation of Section~\ref{S-general-r} and reduce the space needed for storing the array to $N+O(N^{1/r})$ while maintaining constant time access and $O(r)$ amortized time for grow and shrink operations. (Note that this is significant only if $r=r(N)$ is considered to grow with~$N$.) To prove that the $O(r)$ amortized time of resizing operations is optimal, at least for a wide class of data structures that we call \emph{standard}, we introduce in Section~\ref{S-growth-game} an abstract \emph{growth game} and analyze it completely. In Section~\ref{S-lower-grow} we rely on the analysis of the growth game to obtain an $\Omega(r)$ lower bound on the amortized cost of grow operations, if only $N+O(N^{1/r})$ space can be used to store the array.
We end in Section~\ref{S-concl} with some concluding remarks.

\section{Resizable arrays}\label{S-resize}

A resizable array is an abstract data type that supports the following operations:

\smallskip
$A\gets \Array()$ - Create and return an initially empty array.

$A.\Length()$ - Return the current length of the array~$A$.

$A.\Get(i)$ - Return the $i$-th item $a_i$ in the array~$A$. It is assumed that $0\le i< A.\Length()$.

$A.\Set(i,a)$ - Change the $i$-th item in the array~$A$ to~$a$. It is assumed that $0\le i< A.\Length()$.

$A.\Grow(a)$ - Increase the length of array~$A$ by~$1$ and set the new and last item in it to~$a$.

$A.\Shrink()$ - Decrease the length of array~$A$ by~$1$, discarding its last item.

\smallskip
Note that items can only be added to or removed from the end of the array. Allowing items to be added or removed from arbitrary places of the array, updating the indices of the higher-index items accordingly, makes the problem much harder, even if space efficiency is not required. The best time bound that can be obtained simultaneously for all operations is then $O(\log N/\log\log N)$, where~$N$ is the size of the array (Dietz \cite{Dietz89}), and this is tight, as follows from Fredman and Saks \cite{FrSa89}. For any $r\ge 1$, it is possible to implement access operations in $O(r)$ worst-case time, but with insertions and deletions taking $O(N^{1/r})$ amortized time. 
(See, e.g., Goodrich and Kloss \cite{GoKl99}, Joannou and Raman~\cite{JoRa11}, Katajainen \cite{Katajainen16} and Bille et al.\ \cite{BCEG17}.) We stress again that we only allow adding or removing items at the end of the array.

We assume that a memory management system allows us to allocate and deallocate fixed-length arrays of arbitrary size. (As an example, consider the \emph{malloc} and \emph{free} functions of C or C++.) We assume, for simplicity, that each such call takes only constant time. All our amortized bounds hold, however, if allocating or deallocating a fixed-length array of size~$N$ requires $O(N)$ time.
We do not consider the inner workings of the memory management system.

An implementation of a resizable array must use a collection of dynamically allocated fixed-length arrays, sometimes referred to as \emph{blocks}. Typically, each one of these blocks is either a \emph{data block}, containing items of the resizable array, each item in a separate word, or an \emph{index block}, containing pointers to other blocks, each pointer in a separate word, or an \emph{auxiliary block}, containing auxiliary information used by the data structure, such as the lengths of the various blocks. All the data structures presented in this paper are of this form. If the resizable array is currently of size~$N$, then the total size of all the data blocks must be at least~$N$. (See also the discussion after the proof of Theorem~\ref{T-lower1} in Section~\ref{S-lower}.) Each block, except one index block referred to as the \emph{main} index block, must be indicated by a pointer in one of the index blocks, since otherwise the block would be inaccessible. The space used by the data structure is the sum of the lengths of all the blocks allocated.

To simplify the discussion of the various data structures considered in this paper, we introduce the following definition.

\begin{definition}[$(s(N),t(N))$-implementation]\label{D-st}
Let $s(N),t(N)$ be two non-decreasing functions. A resizable array data structure is said to be an \emph{$(s(N),t(N))$-implementation} if it uses at most $N+s(N)$ space to store an array of size~$N$, and at most $N+t(N)$ space during a grow or shrink operation on an array of size~$N$.
\end{definition}

\section{Previous work}\label{S-previous}

The problem of designing space-efficient implementations of resizable arrays belongs of course to the area of \emph{succinct data structures}. See, e.g., Raman et al. \cite{RRR01}, Munro and Srinivasa \cite{MuSr18} and the references therein. But only a handful of papers, which we describe next, seem to address the basic resizable arrays problem as defined here.

\subsection{Basic data structures.}

A \emph{basic} data structure for resizable arrays uses only a single fixed-size array to implement a resizable array. When the fixed-size array is full, a basic data structure needs to allocate a larger array, copy all items from the old array to the new array, add the new item, and release the old array. A basic data structure is free to choose the size of the new fixed-size array. Following some shrink operations, a basic data structure may decide that the fixed-size array is too empty, in which case it can allocate a small array and copy all items into it, releasing the old array.

The most space-efficient basic data structure keeps a fixed-size array whose size is exactly equal to the size of the resizable array. A new array needs to be allocated following each grow and shrink operation. The space used to store an array of size~$N$ is $N+O(1)$. (We need to store the length of the array and a pointer to the fixed-size array.) However, the (amortized) cost of grow and shrink operations is $\Omega(N)$. Furthermore, while implementing a grow or shrink operation, the data structure temporarily needs $2N+O(1)$ space, since two arrays need to stored together. Such a data structure should only be used if grow or shrink operations are rare.

A more practical basic data structure uses the geometric expansion/shrinking technique discussed in the introduction. Choose a fixed parameter $\alpha>0$. When the allocated array is full, allocate a new array of size $(1+\alpha)N$, copy all items into it and release the old array. When an allocated array of size~$N$ contains only $N/(1+\alpha)^2$ items, allocate a new array of size $N/(1+\alpha)$, copy all items into it and release the old array. Simple calculations 
show that the amortized cost of a grow operation is $\frac{1+\alpha}{\alpha}$ and that the amortized cost of a shrink operation is $\frac{1}{\alpha}$. \footnote{More specifically, the amortized cost of a grow operation is $\frac{(1+\alpha)N}{\alpha N}$ while the amortized cost of a shrink operation is 
$\frac{N}{(1+\alpha)^2} \big/ (\frac{N}{1+\alpha}-\frac{N}{(1+\alpha)^2})$. Note that following each resize operation, if~$N$ is the current number of items in the array then the size of the array is $(1+\alpha)N$.
} A judicious choice of $\alpha$ provides a satisfactory solution for many practical situations.

\subsection{The data structure of Sitarski.}\label{sub-HAT}

Sitarski \cite{Sitarski96} (see also \cite{enwiki:1076479510}) described a nice and simple implementation of resizable arrays that he calls a \emph{hashed array tree (HAT)}. This name is a bit unfortunate since no hashing is involved and the tree used by the data structure is actually a list.

The HAT data structure maintains an \emph{index block} $I$ of size~$B$, where $\sqrt{N}\le B<4\sqrt{N}$, and $\lceil N/B\rceil$ or $\lceil N/B\rceil+1$ \emph{data blocks} each of size~$B$, as shown in Figure~\ref{F-HAT}. The~$N$ items currently in the array are stored, in a sequential manner, in the first $\lceil N/B\rceil$ data blocks. The $\lceil N/B\rceil$-th data block is only partially full if~$N$ is not divisible by~$B$. The $(\lceil N/B\rceil+1)$-th data block, if it exists, is empty. The $i$-th entry in the index block contains a pointer to the $i$-th data block.

\begin{figure}[t]
\begin{center}
\includegraphics[scale=0.45]{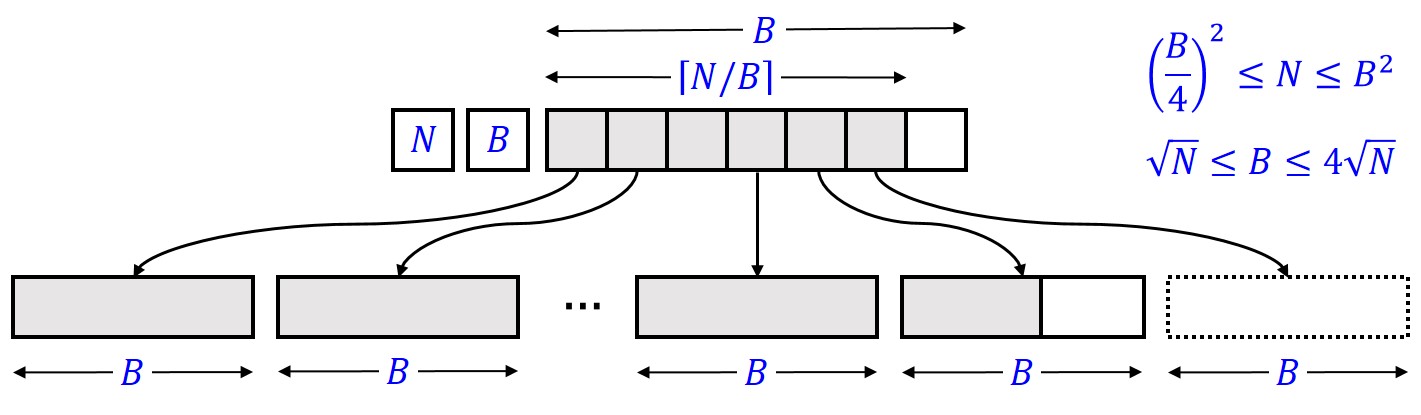} 
\end{center}
\vspace*{-10pt}
\caption{Sitarski's HAT data structure.}\label{F-HAT}
\end{figure}

The data structure uses only $N+3B+2=N+O(\sqrt{N})$ words of memory. (The $3B$ term accounts for the size of the index block and the two possibly empty data blocks.) The $i$-th item is stored at the $(i \bmod B)$-th position of the $\lfloor i/B\rfloor$-th data block, so accessing or modifying it takes $O(1)$ worst-case time. If $B$ is chosen to be a power of~2, computing these indices is especially easy.

If $\sqrt{N}< B$, or equivalently $N<B^2$, then extending the array is simple. If the $\lceil N/B\rceil$-th data block is not full, the new item is added as the last item of this data block. If the $\lceil N/B\rceil$-th data block is full but an $(\lceil N/B\rceil+1)$-st data block is allocated, the new item becomes the first item in this array. Otherwise, a new data block of size $B$ is allocated, and the new item becomes its first item. The extension takes O(1) worst-case time, if allocating a new array is assumed to take $O(1)$ time, or $O(1)$ amortized time, if allocating the new array is assumed to take $O(B)$ time.

When the data structure is full, i.e., $N=B^2$, the data structure is rebuilt with~$B$ doubled. If~$B$ is initially chosen to be a power of~$2$, it will stay a power of~$2$, simplifying the indexing operations. The amortized cost is still $O(1)$ since the $O(N)$ cost of the rebuilt operation can be charged to, say, the last $\frac{3}{4}N$ grow operations that must have occurred since the last rebuild. (Note that if $B=\sqrt{N}$ and $B$ is doubled, then the capacity of the data structure increases from~$N$ to~$4N$.)

The implementation of a shrink operation is similar. To avoid the deallocation and the immediate reallocation of a data block, the last data block is deallocated only if the last two data blocks are empty. When $B=4\sqrt{N}$, or equivalently $N=B^2/16$, the value of~$B$ is halved and the data structure is rebuilt. The amortized cost of grow and shrink operations is still $O(1)$, since at least $\Omega(N)$ operations must occur between two rebuild operations. 

It is important to note that a rebuild operation can be carried out while using only $O(\sqrt{N})$ extra storage. The new and old data blocks are allocated and deallocated one by one.

It is possible to deamortize the HAT data structure by keeping data blocks of two possible sizes and doing a background rebuilding, but the resulting data structure becomes more complicated. 
A simpler data structure with worst-case bounds is described next.

\subsection{The data structures of Brodnik et al.}

Brodnik et al.~\cite{BCDMS99}, apparently unaware of the HAT data structure of Sitarski \cite{Sitarski96}, described an elegant data structure that uses only $N+O(\sqrt{N})$ space and achieves $O(1)$ worst-case time bounds for grow and shrink operations (assuming that memory allocation takes $O(1)$ worst-case time). Appealing features of the data structure are that no rebuild operations are necessary and data items are never moved. 

The data structure of Brodnik et al.~\cite{BCDMS99} comes in two variants, depicted schematically in Figure~\ref{F-Brodnik}. In the first variant, shown in Figure~\ref{F-Brodnik}(a), the items of the array are stored consecutively in data blocks of sizes $1,2,3,\ldots$. As $1+2+\cdots+k=\frac{k(k+1)}{2}\approx \frac{k^2}{2}$, about $\sqrt{2N}$ blocks are needed to store the~$N$ items. More precisely, the number is $k=\lceil\frac{\sqrt{8N+1}-1}{2}\rceil$. The last data block may be only partially filled. An additional, $(k+1)$-st data block may be allocated, in which case it is completely empty. (This last data block once contained items that were subsequently removed.) An index block of size $B=\Theta(\sqrt{N})$ stores pointers to the data blocks. This index block is in fact a resizable array implemented using the na\"ive method.

Growing and shrinking the array are fairly simple operations. A new item is added as the last item of the partially filled block, if there is one, or as the first item of the empty block, if there is one. Otherwise, a new data block is allocated and the new item becomes its first item. A pointer to the new data block is added to the index block. If the index block is full, a larger index block is allocated and all pointers are copied to it. The old index block is deallocated. To obtain an $O(1)$ worst-case time bound, this copying should be done `in the background'. A shrink operation is similar.

The $i$-th item in the array can be accessed in $O(1)$ worst-case time, but calculating the index of the data block in which this item resides requires taking a square root. More precisely, the $i$-th item is stored in the $\ell$-th data block, where $\ell=\lceil\frac{\sqrt{8i+1}-1}{2}\rceil$, and it is the $(i-\frac{\ell(\ell-1)}{2})$-th item in this block.

To avoid the need to compute square roots, Brodnik et al.~\cite{BCDMS99} described a variant of their data structure, shown in Figure~\ref{F-Brodnik}(b). Here the size of each data block is a power of 2. More precisely, for $k=0,1,\ldots$, there is a virtual \emph{super block} composed of $2^{\lfloor k/2\rfloor}$ blocks of size $2^{\lceil k/2\rceil}$. Computing the super block in which the $i$-th block resides, and its position in this super block, can now be done using simple shift operations. For the exact details see \cite{BCDMS99}.\footnote{Actually, there is a small inaccuracy in the details given in \cite{BCDMS99}, but they can easily be fixed.}

\begin{figure}[t]
\centering
\begin{tabular}{lc}
\begin{minipage}[t]{2cm}\vspace*{1cm}(a)\end{minipage} & \includegraphics[scale=0.45,valign=t]{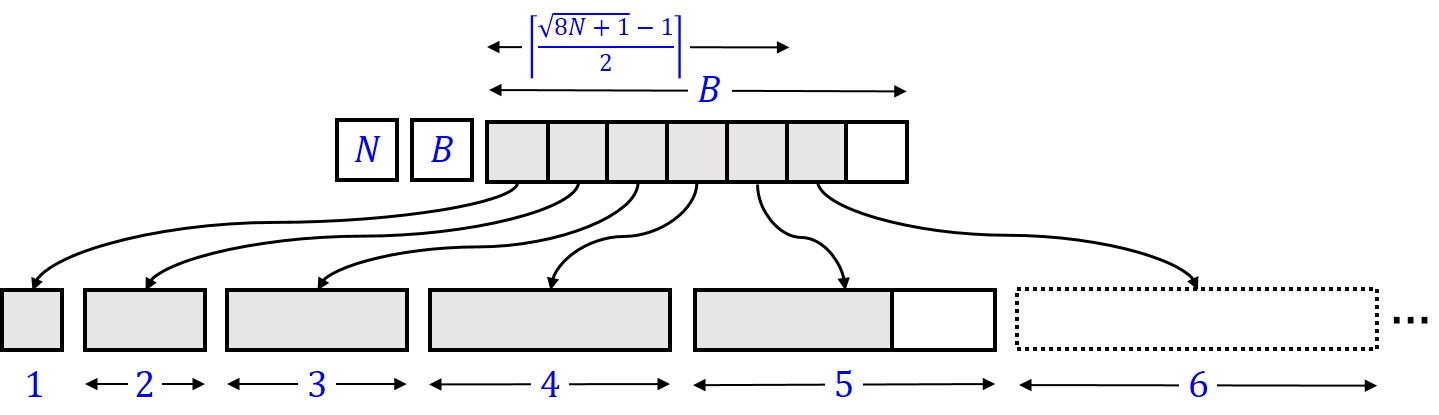} \\ \\[0.1cm]
\begin{minipage}[t]{2cm}\vspace*{1cm}(b)\end{minipage} & \includegraphics[scale=0.45,valign=t]{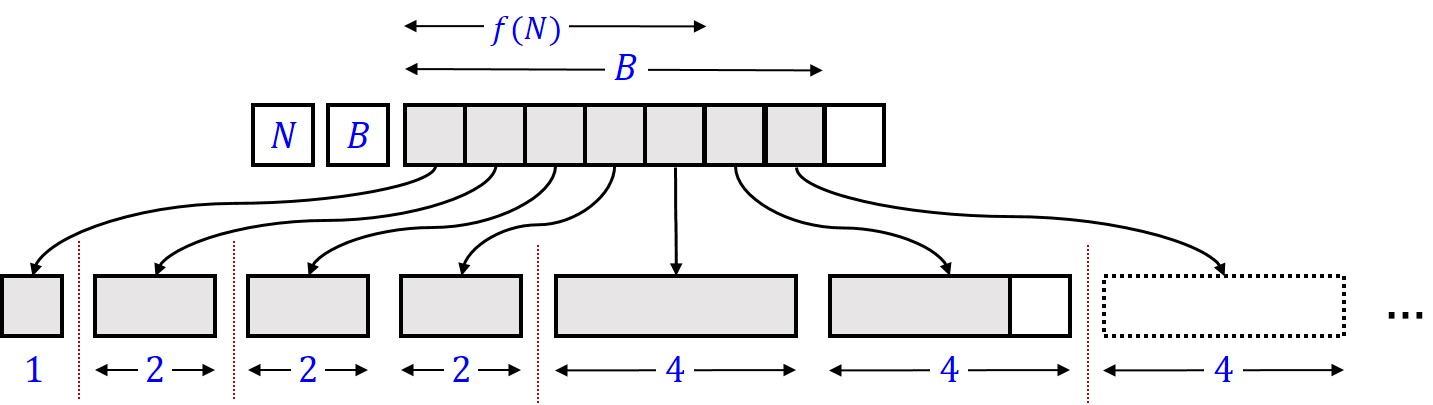} \\ \\[-0.1cm]
\end{tabular}
\caption{Two variants of the data structure of Brodnik et al. (a) The basic version of the data structure. (b) A version in which the sizes of blocks are powers of 2.}\label{F-Brodnik}
\end{figure}

\section{The lower bound of Brodnik et al.\ and its extension}\label{S-lower}

We describe a slightly modified version of a lower bound of Brodnik et al.~\cite{BCDMS99}.

\begin{theorem}\label{T-lower1}
Any data structure for maintaining a resizable array must at certain times use $N+\Omega(\sqrt{N})$ space, where~$N$ is the current length of the array, even if only grow and access operations are performed.
\end{theorem}

\begin{proof}
Let~$k$ be the number of contiguous memory blocks used by the data structure after a sequence of~$N$ grow operations. These blocks must contain all~$N$ items, since otherwise not all items can be accessed. Hence the total size of the blocks is at least~$N$. (See the discussion after the proof.) The data structure must also maintain a pointer to each such block, since otherwise it would not be able to access the block. Thus, if $k\ge \sqrt{N}$, the data structure uses at least $N+k \ge N+\sqrt{N}$ space.

If $k<\sqrt{N}$, then the the largest allocated block must be of size $\ell>\sqrt{N}$. This block is allocated when the array contained $N'\le N$ items. Just after this block was allocated, and its size included in the total space used by the data structure, the $N'$ items must be stored in other blocks. Thus, the total space used at that moment was $N'+\ell>N'+\sqrt{N'}$.
\end{proof}

The lower bound assumes that storing~$N$ items requires~$N$ words of memory. This assumption holds if each word of memory is $w$-bit long and an item is an arbitrary $w$-bit string. The lower bound also assumes that storing~$k$ pointers requires $k$ words of memory. This holds if memory is assumed to be of size $2^w$. If the memory is only of size $2^{cw}$, where $0<c<1$, then only $cw$ bits are required to represent a pointer, and $k$ pointers can be represented using about $ck$ words. This changes the lower bound on the extra space used, but by only a constant factor.

The lower bound of Brodnik et al.~\cite{BCDMS99} can be easily extended to a lower bound that considers separately the space used for storing a resizable array and the temporary space needed while growing or shrinking it. Recall from Definition~\ref{D-st} that an $(s(N),t(N))$-implementation of resizable arrays is an implementation that uses only $N+s(N)$ space to store an array of size~$N$, and at most $N+t(N)$ space during resize operations.

\begin{theorem}
Any $(s(N),t(N))$-implementation of resizable arrays must have $s(N)t(N)\ge N$, even if only grow and access operations are supported.
\end{theorem}

\begin{proof}\label{T-lower2}
Consider the state of the data structure after~$N$ grow operations. All~$N$ items must be spread among a certain number of contiguous memory blocks. Since a pointer must be kept for each such block, and since the extra space used is at most $s(N)$, the number of blocks is at most $s(N)$. Hence at least one of the blocks, call it~$B$, is of size at least $N/s(N)$. Let $N'\le N$ be the index of the grow operation that allocated~$B$. Just after~$B$ was allocated it contained no items. Thus the temporary extra space used by the data structure at that time was at least $n/s(N)$. Thus $t(N)\ge t(N')\ge N/s(N)$, as required. (Note the use of the monotonicity of $t(N)$.)
\end{proof}

\begin{corollary}\label{C-lower2}
For any integer $r\ge 1$, any data structure for that uses only $N+O(N^{1/r})$ space for storing a resizable array of size~$N$ must occasionally use $N+\Omega(N^{1-1/r})$ space during a grow operation, even if only grow and access operations are performed.
\end{corollary}

In the next sections we show that this lower bound is tight.

\section{Simple \texorpdfstring{$(O(N^{1/3}),O(N^{2/3}))$}{(O(N*{1/3}),O(N*{2/3}))}-implementations}\label{S-simple}

In this section we describe simple $(O(N^{1/3}),O(N^{2/3}))$-implementations with $O(1)$ worst-case access time and $O(1)$ amortized time for grow or shrink operations.

\begin{figure}[t]
\centering
\includegraphics[scale=0.45]{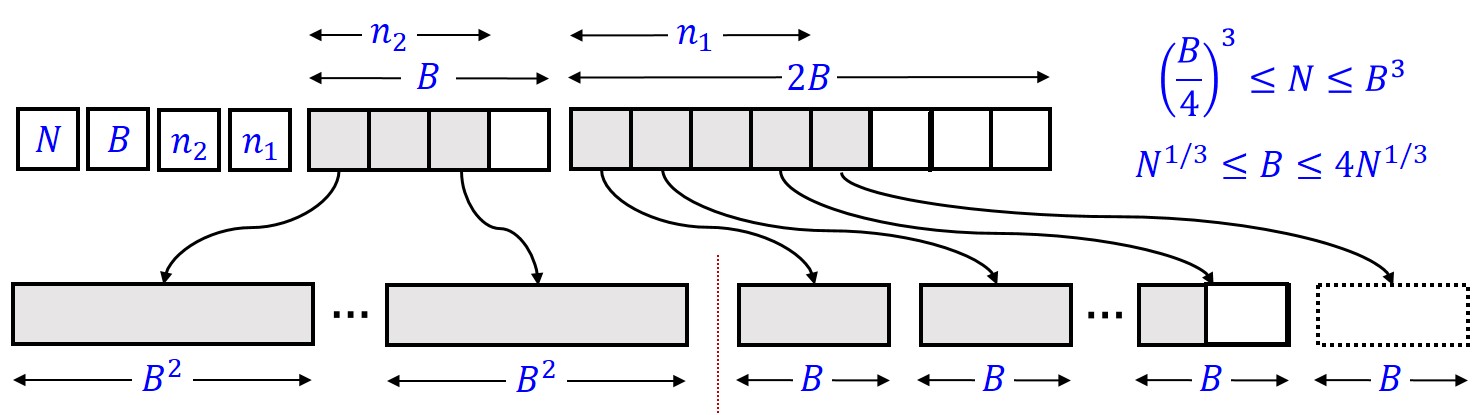} 
\caption{A simple $(O(N^{1/3}),O(N^{2/3}))$-implementation.}\label{F-r3}
\end{figure}

The basic idea is simple. Start by considering the HAT data structure but with data blocks of size roughly~$N^{2/3}$. The length of the index block is roughly $N^{1/3}$. Unfortunately, since the last data block may be almost empty, the total space used by the data structure is $N+O(N^{2/3})$. To fix this, we use a HAT data structure to handle the last partially filled data block of size $N^{2/3}$. The extra space used to store the array is now only $O((N^{2/3})^{1/2})=O(N^{1/3})$. When the size of the last data block reaches $N^{2/3}$, it is copied into a new data block of size $N^{2/3}$. During this copy operation $\Omega(N^{2/3})$ extra space is used. 

It is also easy to obtain a similar data structure based on the data structure of Brodnik et al.~\cite{BCDMS99}. The data structure uses blocks of size $1,4,\ldots,i^2,\ldots,k^2$ and then $1,2,3,\ldots$.

A complete, non-recursive, description of the first data structure is given in Figure~\ref{F-r3}. The data structure maintains a parameter $B$ such that $N^{1/3}\le B\le 4N^{1/3}$. It uses about $N/B^2$ data blocks of size $B^2$, which we call \emph{large} blocks, and at most $2B$ data blocks of size $B$, which we call \emph{small} blocks. There are now two index blocks, one for the large data blocks and one for the small data blocks. Large blocks are always full. At most one small block is partially filled, and at most one small block is empty. When $2B$ small blocks are full, $B$ of them are copied into a new allocated large block. When there are no non-empty small blocks, a large block is split into~$B$ newly allocated small blocks. It is easy to see that all operations require only $O(1)$ time, amortized for grow and shrink, and worst-case for access and modify. A formal proof is given in the next section.

\section{An \texorpdfstring{$(O(rN^{1/r}),O(N^{1-1/r}))$}{}-implementation, for every \texorpdfstring{$r\ge 2$}{}}\label{S-general-r}

Generalizing the construction of Section~\ref{S-simple}, we obtain an $(O(rN^{1/r}),O(N^{1-1/r}))$-implementation, for every $r\ge 2$. It is possible to get such implementations in a recursive manner, as was done in the beginning of Section~\ref{S-simple}. In this section we obtain non-recursive versions of these implementations.
The data structure for a given $r\ge 2$ maintains a parameter~$B$ such that $N^{1/r}\le B< 4N^{1/r}$. It keeps the items in blocks whose sizes are powers of~$B$, i.e., $B,B^2,\ldots,B^{r-1}$. (See Figure~\ref{F-general}.) 

If only grow, and not shrink, operations are to be supported, we can maintain the invariant that there are always at most~$B$ blocks of each size. (To allow efficient shrink operations we shall soon relax this condition and allow at most~$2B$ blocks of each size, as shown in Figure~\ref{F-general}.) Blocks of sizes $B^2,\ldots,B^{r-1}$ are always full. All blocks of length~$B$, except possibly the last one, are also full. If the last block of size~$B$ is not completely full, then a grow operations is easy. If it is full, and there are less than~$B$ blocks of size~$B$, then a new block of size~$B$ is allocated and the new item is placed in it. If there are $B$ full blocks of some size $B^i$, a new block of size $B^{i+1}$ is allocated, the items in these $B$ blocks are copied, in the appropriate order, to the new block, and the $B$ blocks of size $B^i$ are deallocated. A new block of size~$B$ is now allocated and the new item is placed in it. Note that this mimics the operation of a base-$B$ counter.

\begin{figure}[t]
\centering
\includegraphics[scale=0.45]{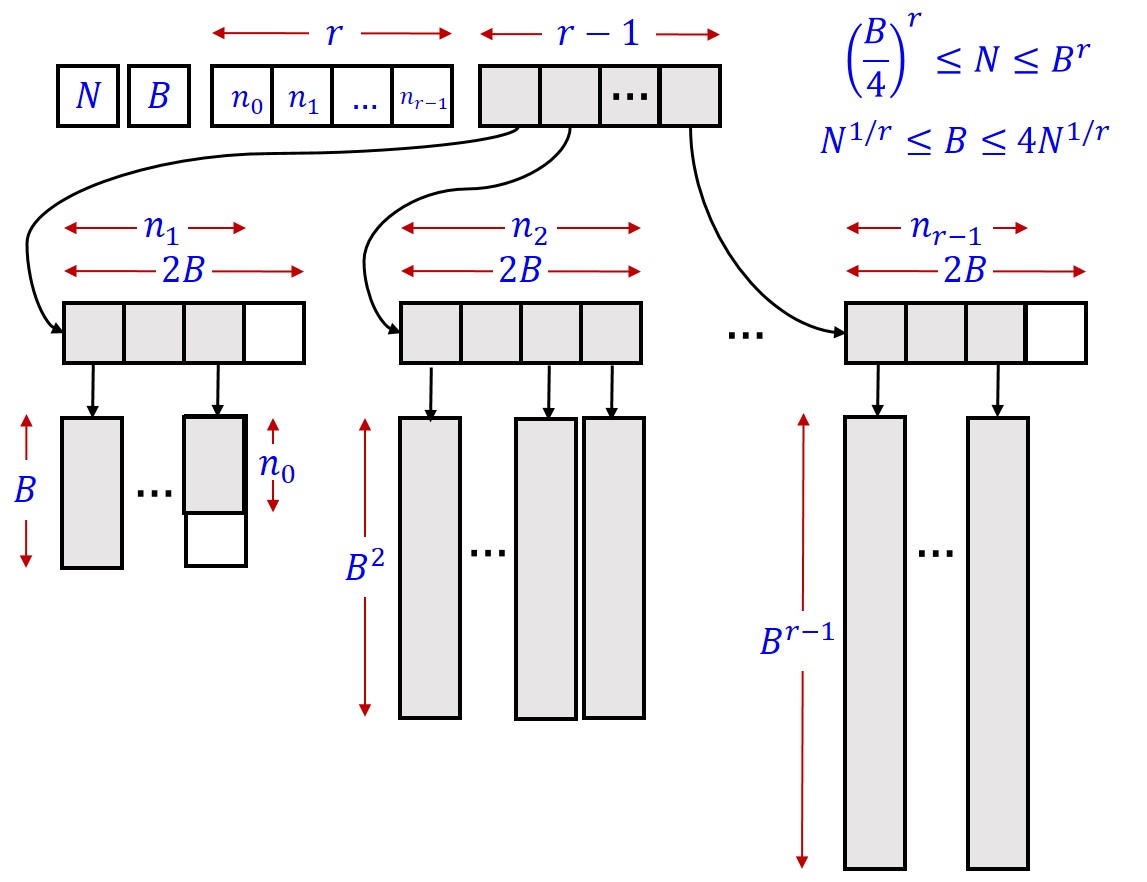} 
\caption{An $(O(rN^{1/r}),O(N^{1-1/r}))$-implementation, for every $r\ge 2$.}\label{F-general}
\end{figure}

When $N=B^r$, the value of~$B$ is doubled and the whole data structure is rebuilt. (We can delay the rebuilding until the data structure is completely full, i.e., $N=B^r+B^{r-1}+\ldots+1=\frac{B^{r+1}-1}{B-1}$, but nothing much is gained by it. For simplicity, we rebuild when $N=B^r$.)

The data structure maintains $r-1$ index blocks, one for each possible block size. The total space used to store an array containing~$N$ items is at most $N+(2r-1)+2(r-1)B+B=N+O(rN^{1/r})$. The maximum amount of extra space used by the implementation is $O(N^{1-1/r})$, when a new block of size $B^{r-1}$ is allocated.

It is easy to see that each item can be accessed in $O(r)$ worst-case time. With some more care, this can be reduced to $O(1)$ worst-case time, as we show in Section~\ref{sub-access}.

The amortized cost of a grow operation is $O(r)$ since each item is first placed in a block of size~$B$ and eventually moved to blocks of size~$B^i$, for $i=2,\ldots,r-1$. A complete rebuilding adds only a constant to the amortized cost of a grow operation since the total number of items in the data structure increases by some constant factor between two rebuild operations.

Allowing shrink operations while maintaining an $O(r)$ amortized cost per grow and shrink operation is not difficult. A standard solution is to use a \emph{redundant} base-$B$ counter. We allow up to $2B$ blocks of each size, not just~$B$. When there are $2B$ full blocks of a given size $B^i$, for $1\le i<r-1$, a new block of size $B^{i+1}$ is allocated, The first~$B$ blocks of size~$B^i$ are copied into it and then deallocated. If a shrink operation occurs when there are no blocks of size~$B$, a block of size~$B^i$, for the smallest possible~$i$, is split into $B-1$ blocks of each size $B^{i-1},\ldots,B^2$ and $B$ blocks of size~$B$. We show in Section~\ref{sub-amortize} that the amortized cost of both grow and shrink operations is again~$O(r)$. A similar redundant base-$B$ counter in used in a different setting by Kaplan et al.~\cite{KTZZ22}.

\subsection{Implementation details.}\label{sub-details}

\begin{figure}[t] 
    \centering
    \parbox{2.9in}{
    \begin{algorithm}[H]
    \Fn{$\Grow(a)$}
    {
        \BlankLine
        \uIf{$N=B^r$}
        { $\Rebuild(2B)$ \;}
        \uElseIf{$n_1=2B$ {\bf and} $n_0=B$}
        {\Combine() \;}
        \ElseIf{$n_1=0$ {\bf or} $n_0=B$}
        {
            $\AAA{1}{n_1}\gets \Allocate(B)$ \;
            $n_1\gets n_1+1$ \; 
            $n_0\gets 0$ \;
        }
        \BlankLine
        $\AAA{1}{n_1-1}[n_0]\gets a$ \;
        $n_0\gets n_0+1$ \;
        $N\gets N+1$ \;
    }
    \end{algorithm}
    } \hspace*{-1cm}
    \parbox{3,8in}{
    \begin{algorithm}[H]
    \Fn{$\Combine()$}
    {
        $k\gets \min\{\, i\in [r-1] \mid n_i<2B\,\}$ \;
        \lIf{$k=\infty$}{\bf error}
        \For{$i\gets k-1$ {\bf downto} $1$}
        {
            $\AAA{i+1}{n_{i+1}}\gets \Allocate(B^{i+1})$ \;
            \For{$j\gets 0$ {\bf to} $B-1$}
            {
                $\Copy(\AAA{i}{j},0,\AAA{i+1}{n_{i+1}},jB^i,B^i)$ \;
                $\Deallocate(\AAA{i}{j})$ \;
                $\AAA{i}{j}\gets \AAA{i}{j+B}$ \tcp*{\rm Shift indices.}
            }
            $n_i\gets B$ \;
            $n_{i+1}\gets n_{i+1}+1$ \;

        }
    }
    \end{algorithm}
    }
    
    \vspace*{10pt}
    \parbox{2.9in}{
    \begin{algorithm}[H]
    \Fn{$\Shrink()$}
    {
        \BlankLine
        \uIf{$N=(B/4)^r$}
        {$\Rebuild(B/2)$ \;}
        \ElseIf{$n_1=0$}
        {\Split() \;}
        \BlankLine
        $n_0\gets n_0-1$ \;
        $N\gets N-1$ \;
        \BlankLine
        \If{$n_0=0$}
        {
            $\Deallocate(\AAA{1}{n_1-1})$ \;
            $n_0\gets B$ \;
            $n_1\gets n_1-1$ \;
        }
    }
    \end{algorithm}
    } \hspace*{-1cm}
    \parbox{3.8in}{
    \begin{algorithm}[H]
    \Fn{$\Split()$}
    {
        $k\gets \min\{\, i\in [r-1] \mid n_i>0\,\}$ \;
        \lIf{$k=\infty$}{\bf error}
        \For{$i\gets k-1$ {\bf downto} $1$}
        {
            $n_{i+1}\gets n_{i+1}-1$ \;
            \For{$j\gets 0$ {\bf to} $B-1$}
            {
                $\AAA{i}{j} \gets \Allocate(B^{i})$ \;
                $\Copy(\AAA{i+1}{n_{i+1}},jB^i,\AAA{i}{j},0,B^i)$ \;
            }
            $\Deallocate(\AAA{i+1}{n_{i+1}})$ \;
        }
    }
    \end{algorithm}
    }
    
    \caption{Pseudocode of grow and shrink operations of the $(O(rN^{1/r}),O(N^{1-1/r}))$-implementation.}
    \label{P-general-r}
\end{figure}

Pseudocode for grow and shrink operations 
is given in Figure~\ref{P-general-r}. As in Figure~\ref{F-general}, we assume that $n_i$, for $i\in[r-1]$, is the number of data blocks of size~$B^i$. We denote these blocks by $\AAA{i}{0},\ldots,\AAA{i}{n_i-1}$. If $n_1>0$, we let~$n_0$ be the number of items in $\AAA{1}{n_1-1}$, the last data block of size~$B$.
With this notation, $A[i]$, for $1\le i\le r-1$, is the $i$-th index block, with size is~$2B$, and~$A$ is the block of pointers to the index blocks, with size~$r$. Indices in each block start from 0. (For convenience, we leave the cell~$A[0]$ empty.) The numbers $n_0,n_1,\ldots,n_{r-1}$ are kept in a block~$n$ of size~$r$. (For simplicity we write $n_i$ instead of~$n[i]$.)
The blocks $A$ and $A[i]$, for $1\le r\le r-1$, are allocated and deallocated only during rebuild operations.

The pseudocode assumes the existence of a procedure $\Allocate(\ell)$ that allocates a new block of size~$\ell$ and returns a pointer to the newly allocated block, and a procedure $\Deallocate(X)$ that deallocates a previously allocated block~$X$. It also assumes that $\Copy(X,x,Y,y,\ell)$ copies the items of $X[x..x+\ell-1]$ to $Y[y..y+\ell-1]$, where $X[x..x+\ell-1]$ is shorthand for $X[x],X[x+1],\ldots,X[x+\ell-1]$. (In other words, it does the assignments $Y[y+i]\gets X[x+i]$, for $i=0,1,\ldots,\ell-1$.) A procedure $\Rebuild(B')$, not shown, rebuilds the data structure with the new parameter $B'$, usually either~$2B$ or~$B/2$.

$\Grow(a)$ works as follows. If $N=B^r$, the data structure is rebuilt with $B\gets 2B$. If the first level is completely full, i.e., $n_1=2B$ and $n_0=B$, $\Combine()$ is called. This procedure finds the smallest index $k\in [r-1]$ for which $n_k<2B$. (Such an index must exist.) For $i\gets k-1,k-2,\ldots,1$, it combines $B$ blocks of size $B^i$ into a new block of size $B^{i+1}$, allocating and deallocating blocks as needed. More specifically, to maintain order, the first $B$ blocks of size $B^i$, i.e., $A[i][0],\ldots,A[i][B-1]$, are combined to form a new last block $A[i+1][n_{i+1}]$ of size~$B^{i+1}$. The blocks $A[i][B],\ldots,A[i][2B-1]$ become blocks $A[i][0],\ldots,A[i][B-1]$. (This only involves pointer changes. No items are moved. It is possible to avoid these pointer changes by considering $A[i]$ to be a cyclic array.) Finally, the number of blocks of size~$B^i$ is set to~$B$ and the number of blocks of size $B^{i+1}$ is increased by one, i.e., $n_i\gets B$ and $n_{i+1}\gets n_{i+1}+1$.

If \Rebuild\ or \Combine\ are not called, $\Grow(a)$ checks whether there is a vacant position in a block of size~$B$. If not, i.e., if $n_1=0$ (no blocks of size~$B$), or $n_1>0$ but $n_0=B$ (all blocks of size~$B$ are full), a new empty block of size~$B$, $A[1][n_1]$, is allocated, $n_1$ is incremented and $n_0$ is set to~$0$. Finally, in all cases, the new item~$a$ is placed in the first vacant position of the empty or partially filled~$B$ block, namely, $A[1][n_1-1][n_0]\gets a$ and~$n_0$ and~$N$ are incremented.

$\Shrink$ is similar. If $N=(B/4)^r$, the data structure is rebuilt with $B\gets B/2$. Otherwise, if $n_1=0$, i.e., there are no blocks of size~$B$, \Split\ finds the smallest $k\in[r-1]$ such that $n_i>0$. The last block of size $B^k$ is split into $B-1$ blocks of each size $B^{k-1},\ldots,B^2$ and into~$B$ blocks of size~$B$. Finally, in all cases, the last item in the last block of size~$B$ is discarded and $n_0$ and~$N$ are decremented. If~$n_0=0$, the empty block of size~$B$ is deallocated, $n_1$ is decremented and $n_1$ is set to~$B$. (It is actually better to delay the deallocation of the last block of size~$B$, as done for $r=2$ in Section~\ref{sub-HAT} and for $r=3$ in Section~\ref{S-simple}. It is not difficult to change the pseudocode accordingly.)

\begin{theorem}\label{T-general-r}
The implementation is an $(O(rN^{1/r}),O(N^{1-1/r}))$-implementation of resizable arrays. The amortized cost of grow and shrink operations is $O(r)$. Access operations can be supported in $O(1)$ worst-case time.
\end{theorem}

\begin{proof}
The correctness of the implementation and the $O(rN^{1/r})$ and $O(N^{1-1/r}))$ space bounds follow from the discussion above. The $O(r)$ bound on the amortized cost of grow and shrink operations is proved in Lemma~\ref{L-amort} below. A way of implementing access operations in $O(1)$ worst-case time is described in the proof of Lemma~\ref{L-access} below.
\end{proof}

By letting $r=\log N$, we get the following interesting corollary:

\begin{corollary}
There is a resizable array implementation that uses $N+O(\log N)$ space to store and array of size~$N$, supporting grow and shrink operations in $O(\log N)$ amortized time and access operations in $O(1)$ worst-case time.
\end{corollary}

\subsection{Amortized cost of grow and shrink operations}\label{sub-amortize}

\begin{lemma}\label{L-amort}
The amortized cost of grow and shrink operations in the $(O(rN^{1/r}),O(N^{1-1/r}))$-implementation is~$O(r)$.
\end{lemma}

\begin{proof}
We define the cost of a \Grow\ or a \Shrink\ operation to be the number of item assignments it performs, and obtain an $O(r)$ amortized bound on these costs. Since the total number of operations performed by a \Grow\ or a \Shrink\ operation is proportional to this defined cost plus one, this is enough to prove the lemma. (The plus one is required to deal with \Shrink\ operations that do no item assignments.)

The cost of a \Grow\ operation is $1$, plus the cost of \Combine, if called. Similarly, the cost of a \Shrink\ operation is~$0$, plus the cost of \Split, if called. (Most \Shrink\ operations only decrease~$N$ and~$n_0$ and do no item assignments.) If a \Combine\ operation creates a block of size~$B^k$, then its cost is $B(B+B^2+\ldots+B^{k-1})\le 2B^k$ (assuming $B\ge 2$). If a \Split\ operation splits a block of size~$B^k$, then its cost is $B^k$.

An item currently in level~$i$, i.e., in a data block of size $B^i$, is assigned $r-i$ \emph{credit} units. Each credit unit can pay for one item assignment. The cost of a \Grow\ operation, excluding the cost of \Combine, if called, is $1$. The credit assigned to the new item is~$r-1$, since it is always placed at level~$1$, so the amortized cost of a \Grow\ operation, excluding the cost of \Combine, is~$r$.

The amortized cost of a \Combine\ operation is~$0$, since each item assignment moves an item from level~$i$ to level~$i+1$, which reduces the credit of the item by~$1$, paying for the assignment. Hence, the amortized cost of \Grow\ is at most~$r$.

To handle \Shrink\ operations we also assign credits to levels. Each \Shrink\ operation adds $3$ units of credit to each level. Hence the amortized cost of \Shrink, excluding the cost of \Split, is at most $3r$.

We next show that the amortized cost of a \Split\ operation is at most~$0$. We first note that at least~$B^k$ \Shrink\ operations must have been performed between two consecutive splits of a block of size~$B^k$. Indeed, following such a split, levels $1$ to~$k-1$ contain exactly $B^k$ items. Furthermore, after a \Combine\ operation that creates a new block of size~$B^k$, levels $1$ to~$k-1$ also contain at least $B^k$ items. Thus, there must be at least $B^k$ \Shrink\ operations between a \Split\ operation at level~$k$ and the previous \Combine\ or \Split\ operation at level~$k$. Similarly, there must be least $B^k$ \Shrink\ operations before the first \Split\ operation at level~$k$. Therefore, when a \Split\ occurs at level~$k$, the level has accumulated at least $3B^k$ unused units of credit.

The actual cost of a \Split\ operation at level~$k$ is $B^k$. The total amount of extra credit that needs to be added to items that are moved from level~$k$ to smaller levels is
\[\textstyle B(k-1)+(B-1)\sum_{i=1}^{k-1} (k-i)B^i \LE B\sum_{i=1}^{k-1}(k-i)B^i \LE B^{k+1}\sum_{j\ge 1}jB^{-j} \EQ \frac{B^{k+2}}{(B-1)^2} \LE 2B^k \;, \]
where the last inequality assumes that $B\ge 4$. Thus, the $3B^k$ units of credit of level~$k$ are sufficient to cover the cost of the \Split\ operation and its amortized cost is therefore at most~$0$.

Finally, we need to account for the cost of rebuilding. Each \Grow\ or \Shrink\ operation now adds $2r$ units of credit to the whole data structure. (This, of course, increases the amortized cost of \Grow\ and \Shrink\ by~$2r$.) The number of \Grow\ operations between a \Rebuild\ operation that doubles~$B$ and the previous \Rebuild\ operation is at least $B^r-(B/2)^r=N(1-2^{-r})N\ge \frac{N}{2}$. The number of \Shrink\ operation between a \Rebuild\ operation that halves~$B$ and the previous rebuild operation is at least $(B/2)^r-(B/4)^r=(2^r-1)N\ge N$. Thus, when a \Rebuild\ is about to take place, the data structure has accumulated at least $rN$ unused credit units. The actual cost of a \Rebuild\ is exactly~$N$. The total amount of credit that needs to be assigned to items is at most $(r-1)N$. Thus, the amortized cost of \Rebuild\ is~$0$.
\end{proof}

The amortized analysis of \Rebuild\ given above is very loose. When a \Rebuild\ doubles~$B$, the total credit of all items is actually decreased. Thus, each \Grow\ operation needs to add only $1/(1-2^{-r})\le 2$ units of credit to the data structure. Similarly, each \Shrink\ needs to add only $r/(2^r-1)\le 1$ units of credit to the data structure. This, of course, does not change the $O(r)$ amortized cost of \Grow\ and \Shrink\ operations.

The analysis above implicitly assumes that~$r$ is a constant, i.e., that it does not vary with~$N$. It is easy to adapt the analysis to the case of non-constant $r=r(N)$. 

\subsection{Accessing and modifying items in \texorpdfstring{$O(1)$}{} worst-case time}\label{sub-access}

\newcommand{\msb}{\mbox{\it msb}}
\newcommand{\lsb}{\mbox{\it lsb}}
\newcommand{\nn}{\bar{n}}
\newcommand{\NNN}{\bar{N}}
\newcommand{\floor}[1]{\lfloor #1 \rfloor}

In this section we show that the $(O(rN^{1/r}),O(N^{1-1/r}))$-implementation of Section~\ref{sub-details} can support access and modify operations in $O(1)$ worst-case time. We assume that the index of the most significant set bit in a machine word~$x$, which we denote by $\msb(x)$, can be determined in $O(1)$ time. (More on this assumption later.)

\begin{lemma}\label{L-access}
The $(O(rN^{1/r}),O(N^{1-1/r}))$-implementation can support access and modify operations in $O(1)$ worst-case time.
\end{lemma}

To simplify the efficient implementation of access operations we make one small changes to the data structure. We let~$n_1$ be the number of full data blocks of size~$B$, i.e., not including the partially filled block of size~$B$ which contains $n_0<B$ items, if there is one. The numbers $n_2,\ldots,n_{r-1}$ are still the number of full blocks of sizes $B^2,\ldots,B^{r-1}$. (All blocks of these sizes are full.) We also assume that~$B$ is a power of~$2$, i.e., $B=2^b$ for some integer $b>0$.

Note that $N=\sum_{j=0}^{r-1}n_jB^j=(n_{r-1},\ldots,n_1,n_0)_B=\sum_{j=0}^{r-1} n_j B^j$, where $0\le n_j\le 2B$, for $0\le j<r$. Thus, $(n_{r-1},\ldots,n_1,n_0)$ is a possibly redundant base-$B$ representation of~$N$. We store $n_{r-1},\ldots,n_1,n_0$ compactly in two words $N^0=(n^0_{n-1},...,n^0_1,n^0_0)_B$ and $N^1=(n^1_{n-1},...,n^1_1,n^1_0)_B$, where $n_j=n^0_j+Bn^1_j$, $0\le n^0_j<B$, $0\le n^1_j\le 2$, for $0\le j< r$. Note that $N^a$, for $a=0,1$, is the concatenation the $b$-bit representations of each one of $n^a_{r-1},\ldots,n^a_1,n^a_0$. Each $n^a_j$ can be extracted from $N^a$ in constant time using simple logical and shift operations. The data structure can easily maintain $N^0$ and $N^1$ in addition to~$N$.

We let $N_k=\sum_{j=0}^{k-1} n_j B^j$, for $0\le k<r$. (We have $N_{r}=N$ and $N_{0}=0$.) Note that $N_k$ is the total number of items in blocks of size at most $B^{k-1}$ and thus $N-N_k$ is the number of items stored in blocks of size at least $B^{k}$. Recall that larger blocks contain items with lower index. Thus, the $i$-th item resides in a block of size at most $B^k$ only if $N-N_{k+1}\le i$. (Recall that indices start from~$0$.)

The numbers $N_k$, for $0\le k<r$ can be explicitly maintained by the data structure, without affecting the asymptotic running times of the various operations. Alternatively $N_k=(N^0\land M_k)+B(N^1\land M_k)$, where $\land$ is the bitwise logical and operation and $M_k$ is a mask of $kb$ $1$'s, i.e., $M_k=2^{kb}-1=(1\ll kb)-1$, where $\ll$ denotes left shift. (The masks can be precomputed.) The multiplication by~$B$ can also be replaced by a left shift of size~$b$, i.e., $N^k=(N^0\land M_k)+(N^1\land M_k)\ll b$.

To access the $i$-th item in the resizable array, where $0\le i<N$, we need to find the unique $0\le k<r$ such that $N-N_{k+1}\le i< N-N_{k}$. The $i$-th item is then in position $(i-(N-N_{k+1})) \bmod B^k$ of the $\lfloor \frac{i-(N-N_{k+1})}{B^k} \rfloor$-th block of size~$B^k$. (If $k=0$ then the $i$-th item is the $(i-(N-N_1))$-th item in the single partially filled block of size~$B$.) Since $B^k$ is a power of~$2$, the $i$-th item can be easily accessed in constant time, using logical and shift operations, once $k$ is known. (Note that $x \bmod B^k = x \land M_k$ and $\floor{x/B^k}=x\gg bk$, where $\gg$ denotes a right shift.)

Let $x=(N-1)-i$ be the index of the $i$-th item from the end of the array, again starting the numbering from~$0$. We then need to find the unique $0\le k<r$ such that $N_{k}\le x < N_{k+1}$.

Let $x=(x_{r-1},\ldots,x_1,x_0)_B$ be the standard base-$B$ representation of~$x$. To find~$k$, we begin by finding the largest index~$\ell$ such that $x_\ell>0$. This can be easily done by finding the index of the most significant set bit in~$x$, i.e., $\ell\gets \floor{\msb(x)/b}$. (We assume again that indices start from~$0$.)

Since $x_\ell>0$ we have $x\ge B^\ell$. Also $N_{\ell-1}\le 2B\sum_{j=0}^{\ell-2}B^j=\frac{2B}{B-1}(B^{\ell-1}-1)\le 3B^{\ell-1}$, assuming $B\ge 4$. Thus, $N_{\ell-1}\le x$ and therefore $k\ge \ell-1$. If $x<N_\ell$ then $k=\ell-1$. Otherwise, if $N_\ell\le x<N_{\ell+1}$ then of course $k=\ell$. If $N_{\ell+1}\le x$, then $k=\ell'$ where $\ell'> \ell$ is the smallest index for which $n_{\ell'}>0$. (Such an index must exist as $x<N$.) This is equivalent to finding the smallest $\ell'> \ell$ for which either $n^0_{\ell'}>0$ or $n^1_{\ell'}>0$. This can be easily done as follows: $\ell'\gets \floor{\lsb((N^0\lor N^1)\land(\neg M_{\ell+1}))/b}$, where $\lsb(x)$ is the index of the least significant set bit in~$x$, $\lor$ is the logical or operation and $\neg$ denotes complementation.

If $x$ is a number, then $x\land (-x)$ leaves only the rightmost $1$ in $x$.
(See Equation (37) in Knuth \cite{knuth2009art}.) Thus, $\lsb$ can be easily reduced to $\msb$, namely $\lsb(x)=\msb(x\land(-x))$. (No reduction in the other direction is known.)

Computing the index of most significant set bit in of a number~$x$ is equivalent to computing the \emph{binary logarithm}, i.e., $\msb(x)=\floor{\lg x}$. Several ways of finding the index of the most significant set bit in a word in constant time are discussed in Knuth \cite{knuth2009art}. One option, if available, is to convert~$x$ to floating point and extract the exponent. Fredman and Willard \cite{FrWi93} describe a clever way of computing $\msb(x)$ using a constant number of logical, shift and arithmetical operations. Their method requires the use of multiplication. The use of multiplication is essential. Brodal, see Knuth \cite{knuth2009art}, devised a way of computing $\msb(x)$ using standard logical, arithmetical and shift operations, without using multiplication, in $O(\log\log w)$ time, where $w$ is the number of bits in a machine word. A matching lower bound was obtained by Brodnik et al.~\cite{BKMM97}.) 

Obtaining a constant access time independent of~$r$ is only of theoretical interest as in practice the value of~$r$ used is expected to be very small. The simple $O(\log r)$ binary search algorithm, or even the na\"{i}ve $O(r)$ algorithm, for locating the value of~$k$ would probably work best.

\section{A data structuring transformation}\label{S-transform}

If~$r$ is a constant, then the $(O(rN^{1/r}),O(N^{1-1/r}))$-implementation of the previous section is optimal, both in the amount of temporary storage used and in its $O(r)=O(1)$ amortized cost of grow and shrink operations. If $r=r(N)$ is a (slowly) growing function of~$N$, it is interesting to ask whether the $O(rN^{1/r})$ extra space needed for storing the array can be reduced to~$O(N^{1/r})$, and whether the amortized cost of grow and shrink operations must be $\Omega(r)$. We may assume that $r(N)=O(\log N)$, since $N^{1/O(\log N)}=O(1)$. 

We show here that the amount of extra space can be reduced to $O(N^{1/r})$, at least if $r(N)\le \frac12 \frac{\log N}{\log \log N}$. In Section~\ref{S-lower-grow} we show that the amortized cost of grow and shrink operations must be $\Omega(r)$. 

We reduce the extra space needed from $O(rN^{1/r})$ to $O(N^{1/r})$ using a surprisingly simple transformation.
To state the transformation, we need the following definition:

\begin{definition}[Standard $(s(N),t(N))$-implementation]\label{D-standard}
An $(s(N),t(N))$-implementation is \emph{standard} if it satisfies the following conditions: \textup{(i)} When a new block~$B$ is allocated, items from some other blocks $A_1,A_2,\ldots,A_k$ are immediately copied into it in sequential order. More specifically, first all the items of~$A_1$ are copied into the first positions of~$B$ and then~$A_1$ is released, then the items of~$A_2$ are copied to the next positions of~$B$, and~$A_2$ is released, and so on. \textup{(ii)} Following each such operation, the total space used by the data structure is $N+s(N)$. 
\end{definition}

In other words, an $(s(N),t(N))$-implementation is standard if the $t(N)$ temporary extra space is only briefly needed following the allocation of new blocks. It is not difficult to check that all implementations given in the paper are standard. We can now obtain the following result.

\begin{lemma}\label{L-transform}
A standard $(s(N),O(N))$-implementation can be converted into an $(s(N)+O(\frac{N}{t(N)}),\allowbreak s(N)+O(t(N)))$-implementation, for any non-decreasing function $t(N)$, without increasing the asymptotic worst-case cost of access operations and the asymptotic amortized cost of grow and shrink operations.
\end{lemma}

\begin{proof}
Given the definition of standard implementations, we only need to consider the allocation of a new block~$B$ that is to receive the items currently stored in $A_1,A_2,\ldots,A_k$. Let~$b$ be the length of~$B$. If $b\le t(N)$, there is no problem. Otherwise, instead of allocating~$B$, we sequentially allocate and fill $\lceil b/t(N)\rceil$ blocks $B_1,B_2,\ldots,B_{\lceil b/t(N)\rceil}$ of size $t(N)$. We of course need to keep pointers to these blocks. Since the total size of all blocks allocated at any given time is at most $N+s(N)$, the number of extra pointers needed is at most $(N+s(N))/t(N)$. The total space needed to store an array of size~$N$ is at most $N+s(N)+O(\frac{N}{t(N)})$. The total space needed at any time by the data structure is at most $N+s(N)+O(\frac{N}{t(N)})+t(N)$, which is $s(N)+O(t(N))$ since we may assume that $t(N)\ge \sqrt{N}$.
\end{proof}

Using the lemma, we easily get the following theorem:

\begin{theorem}\label{T-r}
For any integer $r=r(N)\le \frac12 \frac{\log N}{\log \log N}$, there exists an $(O(N^{1/r}),O(N^{1-1/r}))$-implemen\-ta\-tion that supports access and modify operations in $O(1)$ worst-case time and grow and shrink operations in $O(r)$ amortized time.
\end{theorem}

\begin{proof}
Apply the transformation of Lemma~\ref{L-transform}, with $t(N)=N^{1-1/r}$, to the standard $(O(rN^{1/(2r)}),\allowbreak O(N^{1-1/(2r)}))$-imple\-men\-ta\-tion of Section~\ref{S-general-r}. The result is an $(O(rN^{1/(2r)}+N^{1/r}),O(N^{1-1/r}))$-implementation. If $r\le \frac12 \frac{\log N}{\log \log N}$ then $2r\log r\le \log n$ and hence $rN^{1/(2r)}\le N^{1/r}$. Thus the obtained implementation satisfies the requirements of the theorem.
\end{proof}

\section{The growth game}\label{S-growth-game}

To get an indication as to how a truly optimal data structure should handle grow operations, and to obtain an~$\Omega(r)$ lower bound on the amortized cost of the grow operations of a standard data structure that is only allowed $N+O(N^{1/r})$ space to store an array of size~$N$, we define and analyze an interesting solitaire, i.e., a one-player, game called the \emph{growth} game, which is an abstraction of the way standard data structures can handle grow operations. (See Definition~\ref{D-standard}.)

\subsection{The \texorpdfstring{$(N,k,\ell)$}{}-growth game.}\label{sub-game}

Let $N,k$ and $\ell$ be fixed parameters. We are required to insert~$N$ items, one by one, into~$k$ initially empty subarrays $A_1,A_2,\ldots,A_k$ that represent a resizable array~$A$. The first items of~$A$ are stored in $A_k$, the next in $A_{k-1}$, and so on. (This apparent reversal turns out to simplify things.)  Let $a_1,a_2,\ldots,a_k$ be the number of items in $A_1,A_2,\ldots,A_k$, respectively.

Some of the subarrays $A_1,A_2,\ldots,A_k$ may be \emph{empty}, i.e., contain no items. (If~$A_i$ is empty then $a_i=0$.) Empty subarrays, if any, are assumed to be the first subarrays, i.e., if $a_i=0$ then $a_1=a_{2}=\cdots=a_i=0$. The first nonempty subarray may contain up to~$\ell$ \emph{vacant} positions. All other subarrays are \emph{full}, i.e., the number of items stored in them is equal to the size of the memory block allocated for them. The extra space used to store the items is thus always at most $k+\ell$.

A grow operation, which adds a new item to the array, is implemented as follows. If the first non-empty subarray has vacant positions, the item is simply placed in its first vacant position. (See Figure~\ref{F-growth-game}(a).) If all non-empty subarrays are full, but there is at least one empty subarray, a new subarray of size $\ell+1$ is allocated and the new item is placed in its first position. The remaining~$\ell$ positions are left vacant. This subarray replaces the last empty subarray. (See Figure~\ref{F-growth-game}(b).) The cost of inserting the new item in both these cases is defined to be 1. (For simplicity, we ignore the cost of allocation.) Note that in these first two cases the player of the game has no real choice. 

The game is more interesting when all~$k$ subarrays are full. We now have~$k$ possible moves. For each $i\in [k]=\{1,2,\ldots,k\}$, we can allocate a new subarray of size $(\ell+1)+\sum_{j=1}^i a_j$, copy the items in $A_i,A_{i-1},\ldots,A_1$, in this order, to the new subarray, and then place the new item in the first vacant position of the new subarray. (The new subarray will still have~$\ell$ vacant positions.) The old subarrays $A_1,\ldots,A_i$ are deallocated. The new subarray becomes $A_i$ while $A_1,A_2,\ldots,A_{i-1}$ become empty. (See Figure~\ref{F-growth-game}(c).) The cost of the operation is defined to be $1+\sum_{j=1}^i a_j$. (Again we count only item assignments and ignore the cost of allocations and deallocations.) A move of the second type is actually a move of the third type, but~$i$ is forced to be the largest index for which $a_i=0$.

\begin{figure}[t]
\centering
\begin{tabular}{rrrrrr}
\includegraphics[scale=0.45,valign=b]{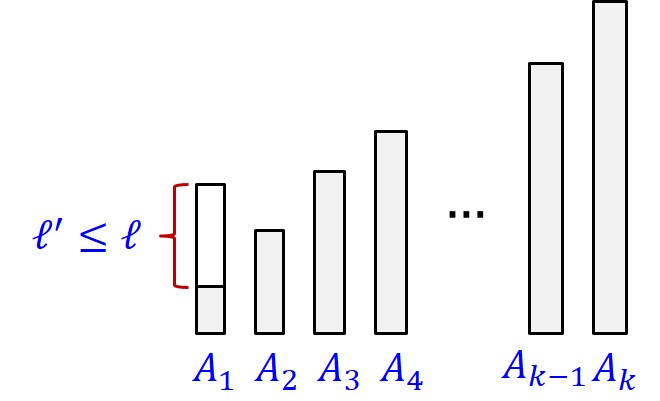} & &
\includegraphics[scale=0.45,valign=b]{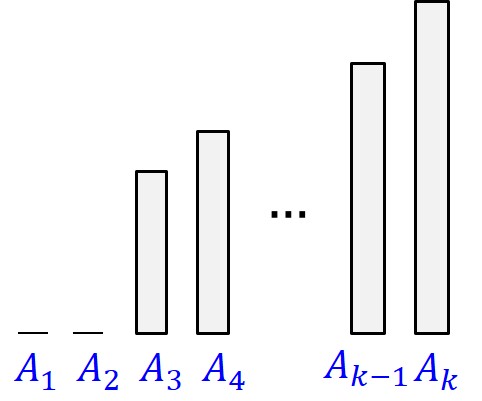} & &
\includegraphics[scale=0.45,valign=b]{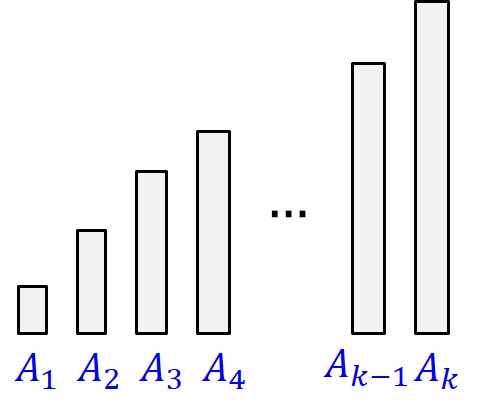} \\ \\[-0.15cm]
\multicolumn{1}{c}{$\big\Downarrow$} & & \multicolumn{1}{c}{$\big\Downarrow$} & & \multicolumn{1}{c}{$\big\Downarrow$} \\ \\[-0.3cm]
\includegraphics[scale=0.45,valign=b]{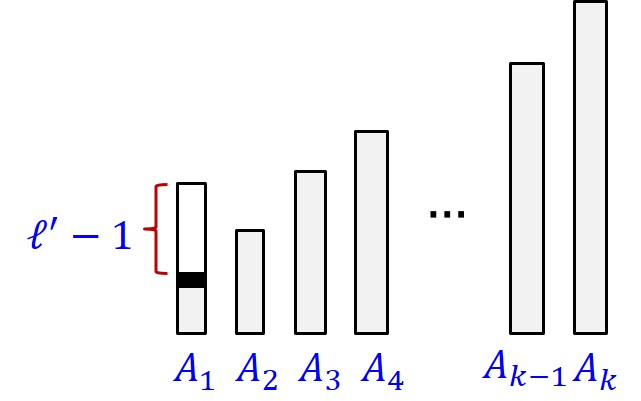} & &
\includegraphics[scale=0.45,valign=b]{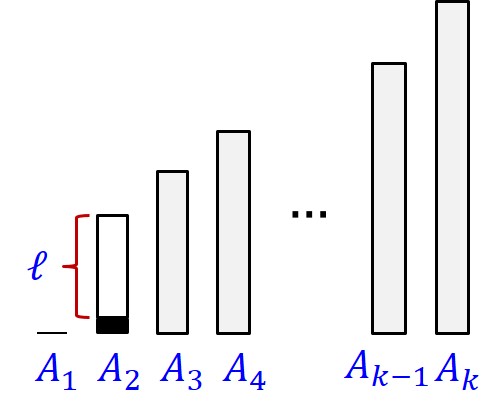} & &
\includegraphics[scale=0.45,valign=b]{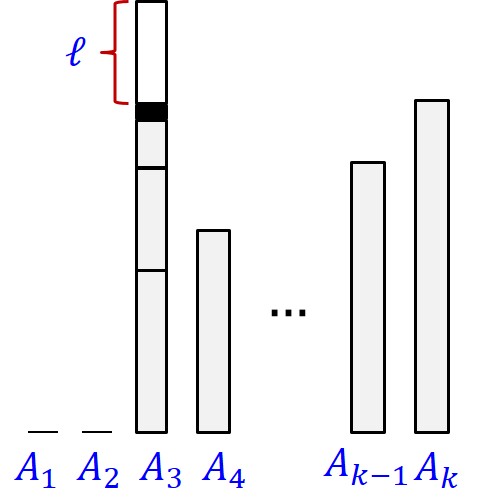} \\ \\[-0.3cm]
\multicolumn{1}{c}{(a)} & & \multicolumn{1}{c}{(b)} & & \multicolumn{1}{c}{(c)}
\end{tabular}
\caption{Possible moves in the growth game. (a) If the first non-empty subarray has vacant positions, a new item is inserted into its first vacant position. The cost of the operation is~$1$. (b) If there are no vacant positions, but there is at least one empty subarray, then a new array of size $\ell+1$ is allocated. (c) If there are no vacant positions, the subarrays $A_1,A_2,\ldots,A_i$, for some $i\in[k]$, are merged into a newly formed $A_i$, initially with $\ell+1$ vacant positions. The new item is placed in the first vacant position, leaving $\ell$ vacant positions. The subarrays $A_1,A_2,\ldots,A_{i-1}$ become empty. The cost of the operation is $1+\sum_{j=1}^i a_j$. (In the example, $i=3$.)}\label{F-growth-game}
\end{figure}

\begin{definition}[Total and amortized cost]  Let $C_{N,k,\ell}$ be the minimum total cost required to play the $(N,k,\ell)$-growth game. Let $A_{N,k,\ell}=\frac{C_{N,k,\ell}}{N}$ be the corresponding amortized cost of a single grow operation.
\end{definition}

We are interested in finding $C_{N,k,\ell}$, and hence $A_{N,k,\ell}$, and an optimal sequence of moves that achieves these values. The solution to these problems turns out to be fairly simple, and also fairly interesting.

The values of $A_{N,cN^{1/r},cN^{1/r}}$, for some constant $c>0$, can be used, as we shall see, to bound from below the amortized cost of a grow operation when only $cN^{1/r}$ extra space is allowed. We are more liberal here, since we assume that~$N$ is known in advance and we allow $cN^{1/r}$ extra space even when the array contains many fewer than~$N$ items. We also do not impose any bound on the amount of space used temporarily during grow operations. All this only makes our lower bound stronger.

\subsection{Reduction to the case \texorpdfstring{$\ell=0$}{}.}\label{sub-l=0}

It might seem that allowing up to $\ell$ vacant positions in one of the subarrays is an important feature of the game that helps reduce the amortized cost of growth operations. We begin by showing, perhaps surprisingly, that vacant positions do not help much, at least for most values of the parameters.
More specifically, we describe a reduction from the $\ell>0$ case to the $\ell=0$ case.

\begin{lemma}
For every~$N$ divisible by $\ell+1$, $C_{N,k,\ell}=(\ell+1)C_{\frac{N}{\ell+1},k,0}$ and $A_{N,k,\ell}=A_{\frac{N}{\ell+1},k,0}$.
\end{lemma}

\begin{proof}
Playing the $(N,k,\ell)$-growth game amounts to playing the $(\frac{N}{\ell+1},k,0)$-growth game on blocks of $\ell+1$ items. To see this, note that when a grow operation in the $(N,k,\ell)$-growth game leaves~$\ell$ vacant positions, the next~$\ell$ grow operations must immediately fill these positions.
\end{proof}

\subsection{Analysis of the \texorpdfstring{$\ell=0$}{} case.}

Given the simple reduction of Section~\ref{sub-l=0}, it is enough to consider the case $\ell=0$. To simplify the notation, let $C_{N,k}=C_{N,k,0}$ and $A_{N,k}=A_{N,k,0}$. We refer to the corresponding game as the $(N,k)$-growth game. Let $\NN=\{0,1,\ldots\}$ be the set of non-negative integers. 

\begin{definition}[States and their cost]\label{D-C}
Let 
\[ P_{N,k} \EQ \left\{\, \aaa=(a_1,a_2,\ldots,a_k)\in \NN^k \;\middle|\; \sum_{i=1}^k a_i=N \text{ \rm and } \,a_i=0\, \Rightarrow \,a_{i-1}=0\, \;,\; i=2,\ldots,k \,\right\} \;, \]
be the set of all states of total size~$N$ in the $(N,k)$-growth game. For $\aaa\in P_{N,k}$, let $C(\aaa)=C_k(\aaa)$ be the minimum total cost needed to reach state $\aaa=(a_1,a_2,\ldots,a_k)$, i.e., $|A_i|=a_i$, for $i\in[k]$, starting from state $(0,0,\ldots,0)$. Clearly $C_{N,k}=\min\{\, C(\aaa) \mid \aaa\in P_{N,k} \}$.
\end{definition}

The following lemma gives simple recurrence relations for computing the costs of states. These relations by themselves do not provide an efficient way of computing $C_{N,k}$ and $C(\aaa)$, for every $\aaa\in P_{N,k}$. They allow us to obtain explicit formulas for these quantities, however.

\begin{lemma}\label{L-decompose}
For every $\aaa=(a_1,a_2,\ldots,a_k)\in P_{N,k}$: 
\begin{itemize}[topsep=0pt,left=15pt,itemsep=2pt,parsep=0pt]
\item[\textup{(i)}] $C_k(a_1,a_2,\ldots,a_k) = C_k(0,\ldots,0,a_k)+C_{k-1}(a_1,a_2,\ldots,a_{k-1})$.
\item[\textup{(ii)}] $C_k(a_1,a_2,\ldots,a_k) =\sum_{j=1}^k C_j(0,\ldots,0,a_j)$.
\item[\textup{(iii)}] $C_k(0,\ldots,0,a_k) = C_{a_k-1,k}+a_k$.
\item[\textup{(iv)}] $C(\aaa)=N+\sum_{j=i}^k C_{a_j-1,j}$, if $0=a_{i-1}<a_i$.
\end{itemize}
\end{lemma}

\begin{proof}
Let $(a_1,a_2,\ldots,a_k)\in P_{N,k}$. Consider a sequence of moves from $(0,0,\ldots,0)$ to $(a_1,a_2,\ldots,a_k)$. The last move to change $A_k$ must be to the state $(0,\ldots,0,a_k)$. The minimum cost of reaching this state is $C_k(0,\ldots,0,a_k)$. The minimum cost of moving from state $(0,\ldots,0,a_k)$ to $(a_1,a_2,\ldots,a_k)$ is $C_{k-1}(a_1,a_2,\ldots,a_{k-1})$, since only $A_1,A_2,\ldots,A_{k-1}$ can be used. This establishes (i); (ii) follows by induction.

The move to $(0,\ldots,0,a_k)$ must be from a state $(b_1,b_2,\ldots,b_k)$ with $\sum_{j=1}^{k}b_j=a_k-1$. The minimum cost of reaching such a state is exactly $C_{a_k-1,k}$ and the cost of the move to $(0,\ldots,0,a_k)$ is~$a_k$. This establishes (iii); (iv) follows by induction.
\end{proof}

The exact formula that we obtain for $C_{N,k}$ involves binomial coefficients. It is thus convenient to start by stating two well known binomial-coefficient identities.

\begin{lemma}\label{L-binom} For every $n,k\ge 0$:
\textup{(i)} $\binom{n}{k}=\binom{n-1}{k-1}+\binom{n-1}{k}$ and 
\textup{(ii)} $\sum_{i=0}^{k}\binom{n+i-1}{i} = \binom{n+k}{k}$.
\end{lemma}

\begin{proof} (i) is well known. (It may even be used as the definition of the binomial coefficients, along with $\binom{0}{0}=1$.) (ii) is obtained by an easy induction on~$k$. For $k=0$, both terms are~$1$. (Note that $\binom{n}{0}=1$ for every~$n$.) The induction step follows easily:
\[ \textstyle \sum_{i=0}^{k}\binom{n+i-1}{i} \EQ \left(\sum_{i=0}^{k-1}\binom{n+i-1}{i}\right) + \binom{n+k-1}{k} \EQ \binom{n+k-1}{k-1} + \binom{n+k-1}{k} \EQ \binom{n+k}{k} \;. \qedhere
\]
\end{proof}

The following theorem gives an exact formula for computing the total cost $C_{N,k}$ and a characterization of all the states in~$P_{N,k}$ that achieve this minimum total cost. (The corresponding optimal sequences of moves are considered in Section~\ref{sub-play}.)


\begin{theorem}\label{T-main}
\textup{(i)} If $\binom{n+k-1}{k}-1\le N \le \binom{n+k}{k}-1$, for some $n\ge 0$, then
\[ \textstyle C_{N,k} \EQ (N+1)n - \binom{n+k}{k+1} \;.\] 
\textup{(ii)} If $\binom{n+k-1}{k}\le N < \binom{n+k}{k}$, for some $n\ge 1$, then
\[ C_{N,k}-C_{N-1,k} \EQ n \;.\]
\textup{(iii)} If $\binom{n+k-1}{k}-1\le N < \binom{n+k}{k}-1$, for some $n\ge 0$, and $\aaa\in P_{N,k}$, then 
\[ \textstyle C_{N,k}\EQ C(\aaa) 
\quad\Longleftrightarrow\quad
\binom{n+i-2}{i}\le a_i \le \binom{n+i-1}{i}\;, \text{ \rm  for every } i\in[k]\;.  \]
\end{theorem}

Both inequalities in claim (i) are weak inequalities. Thus if $N=\binom{n+k}{k}-1$ for some $n\ge 0$, then condition (i) is satisfied by both~$n$ and $n+1$. This is consistent since for this value of~$N$,
\begin{align*}
   & \textstyle (N+1)(n+1) - \binom{n+k+1}{k+1} \EQ (N+1)n + (N+1-\binom{n+k+1}{k+1}) \\
  \EQ & \textstyle (N+1)n + (\binom{n+k}{k+1} -\binom{n+k+1}{k+1}) \EQ (N+1)n -\binom{n+k}{k+1} \;,
\end{align*}
where the last equality follows by Lemma~\ref{L-binom}(i).

Theorem~\ref{T-main}(iii) can be stated more succinctly using the following definition:

\begin{definition}[$Q_{N,k}$]\label{D-Q}
If $\binom{n+k-1}{k}\le N < \binom{n+k}{k}$, let
\[ \textstyle Q_{N,k} \EQ \left\{ \aaa=(a_1,a_2,\ldots,a_k)\in P_{N,k} \;\middle|\; \binom{n+i-2}{i}\le a_i \le \binom{n+i-1}{i}\;, \text{ \rm  for every } i\in[k]\right\}\;.\]
\end{definition}

Theorem~\ref{T-main}(iii) says that if $\aaa\in P_{N,k}$, then $C_{N,k}=C(\aaa)$ if and only if $\aaa\in Q_{N,k}$. That is, $Q_{N,k}$ is the set of optimal final positions in the $(N,k)$-growth game. For most values of~$N$ there are many optimal final positions. If $N=\binom{n+k}{k}-1$, then $(\binom{n}{1},\binom{n+1}{2},\ldots,\binom{n+k-1}{k})$ is the only optimal final state and the total cost is $\binom{n+k}{k}n-\binom{n+k}{k-1}=\frac{kn}{k+1}\binom{n+k}{k}$. (We say more about this in Section~\ref{sub-play}.)

We prove Theorem~\ref{T-main} by induction on~$N$ and $k$. We say that $(N',k')<(N,k)$ if $N'<N$ and $k'\le k$. We say that $(N',k')\le (N,k)$ if $N'\le N$ and $k'\le k$. We break the proof into small pieces, each stated as a separate claim. The basis of the induction is established by the following claim:

\begin{claim}\label{C-basis}
Theorem~\textup{\ref{T-main}} holds for $n=1$, i.e., when $0\le N\le k$.
\end{claim}

\begin{proof}
If $0\le N\le k$, then $C_{N,k}=N$, i.e., the insertion of each item costs exactly~$1$, and if $\aaa\in P_{N,k}$, then $C(\aaa)=N$ if and only if $0\le a_i\le 1$, for every $i\in [k]$.
\end{proof}

Before establishing the induction step we prove that, in a sense, Theorem~\ref{T-main}(iii) implies Theorem~\ref{T-main}(i), and that Theorem~\ref{T-main}(i) implies Theorem~\ref{T-main}(ii). 

\begin{claim}\label{C-diff}
If Theorem~\textup{\ref{T-main}(i)} holds for $(N,k)$ and $(N-1,k)$, then Theorem~\textup{\ref{T-main}(ii)} holds for $(N,k)$, i.e., if $\binom{n+k-1}{k}\le N < \binom{n+k}{k}$, for some $n\ge 1$, then $C_{N,k}-C_{N-1,k}=n$.
\end{claim}

\begin{proof}
If $\binom{n+k-1}{k}\le N < \binom{n+k}{k}$, for some $n\ge 1$, then 
\[ \textstyle \binom{n+k-1}{k}-1 \LT N \LE \binom{n+k}{k}-1 \quad,\quad \binom{n+k-1}{k}-1\LE N-1 \LT \binom{n+k}{k}-1 \;.\]
Theorem~\ref{T-main}(i) applied to $(N,k)$ and $(N-1,k)$ gives $C_{N,k}-C_{N-1,k}=n$. (Note that here we use the fact that Theorem~\ref{T-main}(i) requires only the weak inequality $N \le \binom{n+k}{k}-1$.)
\end{proof}

\begin{claim}\label{C-CNk}
If Theorem~\textup{\ref{T-main}(i)} holds for every $(N',k')< (N,k)$, $\binom{n+k-1}{k}-1\le N < \binom{n+k}{k}-1$, for some $n\ge 1$, and $\aaa\in Q_{N,k}$, i.e., $\aaa\in P_{N,k}$ and $\binom{n+i-2}{i}\le a_i \le \binom{n+i-1}{i}$ for every $i\in[k]$, then $C(\aaa)=(N+1)n - \binom{n+k}{k+1}$. (Note that this is the claimed value of $C_{N,k}$.)
\end{claim}

\begin{proof}
Suppose that $\aaa\in P_{N,k}$ and that $\binom{n+i-2}{i}\le a_i \le \binom{n+i-1}{i}$, for every $i\in[k]$. Then,
\begin{align*}
    C(\aaa) & \textstyle \EQ  N + \sum_{i=1}^k C_{a_i-1,i} \\
    & \textstyle \EQ N + \sum_{i=1}^k (a_i(n-1)-\binom{n+i-1}{i+1}) \\
    & \textstyle \EQ (N+1)n - ( 1 + (n-1) + \sum_{i=1}^k \binom{n+i-1}{i+1} ) \\
    & \textstyle \EQ (N+1)n - \sum_{j=0}^{k+1} \binom{n+j-2}{j} \EQ (N+1)n - \binom{n+k}{k+1} \;,
\end{align*}
In the first line we used Lemma~\ref{L-decompose}, using the fact that $a_i\ge 1$ for every $i\in[k]$. In the second line we used Theorem~\ref{T-main}(i) for $(a_i-1,i)<(N,k)$. The third line is obtained by a simple rearrangement. The fourth line follows by Lemma~\ref{L-binom}. Note that 
\[ \textstyle \sum_{j=0}^{k+1} \binom{n+j-2}{j} \EQ \sum_{j=0}^{k+1} \binom{(n-1)+j-1}{j} \EQ \binom{(n-1)+(k+1)}{k+1} \EQ \binom{n+k}{k+1}\; . 
\qedhere 
\]
\end{proof}

\begin{claim}\label{C-ij}
Suppose that $\binom{n+k-1}{k}-1\le N \le \binom{n+k}{k}-1$, for some $n\ge 1$, and that $\aaa\in P_{N,k}$.  \\
\textup{(i)} If $a_i<\binom{n+i-2}{i}$ for some $i\in [k]$, then there exists $j\in[k]$ such that $a_j>\binom{n+j-2}{j}$. \\
\textup{(ii)} If $a_j>\binom{n+j-1}{j}$ for some $j\in [k]$, then there exists $i\in[k]$ such that $a_i<\binom{n+i-1}{i}$. 
\end{claim}

\begin{proof}
(i) If $a_j\le \binom{n+j-2}{j}$ for every $j\in [k]$, then 
\[ \textstyle N \EQ \sum_{j=1}^k a_j \LT \sum_{j=1}^k \binom{n+j-2}{j} \EQ \binom{n+k-1}{k}-1 \;, \]
contradicting the definition of~$n$.

(ii) Similarly, if $a_i\ge \binom{n+i-1}{i}$ for every $i\in[k]$, then 
\[ \textstyle N \EQ \sum_{i=1}^k a_i \GT \sum_{i=1}^k \binom{n+i-1}{i} \EQ \binom{n+k}{k}-1 \;, \]
again contradicting the definition of~$n$.
\end{proof}

\begin{claim}\label{C-non-opt}
If Theorem~\textup{\ref{T-main}(ii)} holds for every $(N',k')< (N,k)$, $\binom{n+k-1}{k}-1\le N < \binom{n+k}{k}-1$, for some $n\ge 1$, and $\aaa\in P_{N,k}\setminus Q_{N,k}$, i.e., $\aaa\in P_{N,k}$ but $\binom{n+i-2}{i}\le a_i \le \binom{n+i-1}{i}$ is \emph{not} satisfied for some $i\in[k]$, then there exists $\bbb\in P_{N,k}$ such that $C(\aaa)>C(\bbb)$. As a consequence, $C(\aaa)>C_{N,k}$.
\end{claim}

\begin{proof}
There are two (non-exclusive) cases:

{\bf Case 1:} There exists $i\in [k]$ for which $a_i<\binom{n+i-2}{i}$.

By Claim~\ref{C-ij}(i), there exists $j\in [k]$ such that $a_j>\binom{n+j-2}{j}$. Let $\bbb=(b_1,b_2,\ldots,b_k)\in P_{N,k}$ be such that $b_i=a_i+1$, $b_j=a_j-1$, and $b_\ell=a_\ell$ for $\ell\in [k]\sm\{i,j\}$. As Theorem~\ref{T-main}(ii) is assumed to hold for $(a_i,i)$ and $(a_{j}-1,j)$, we have $C_{a_i,i}-C_{a_i-1,i} \le n-2$ while $C_{a_j-1,i}-C_{a_j-1,i} \ge n-1$. Thus
\[ C(\bbb)-C(\aaa) \EQ (C_{a_i,i}-C_{a_i-1,i}) - (C_{a_j-1,j}-C_{a_j-2,j}) \LT 0\;.\]

{\bf Case 2:} There exists $j\in [k]$ for which $a_j>\binom{n+j-1}{j}$.

By Claim~\ref{C-ij}(ii), there exists $i\in[k]$ such that $a_i<\binom{n+i-1}{i}$. Again, let $\bbb=(b_1,b_2,\ldots,b_k)\in P_{N,k}$ be such that $b_i=a_i+1$, $b_j=a_j-1$, and $b_\ell=a_\ell$ for $\ell\in [k]\sm\{i,j\}$. As Theorem~\ref{T-main}(ii) is assumed to hold for $(a_i,i)$ and $(a_{j}-1,j)$, we have $C_{a_i,i}-C_{a_i-1,i} \le n-1$ while $C_{a_j-1,i}-C_{a_j-1,i} \ge n$. Thus,
\[ C(\bbb)-C(\aaa) \EQ (C_{a_i,i}-C_{a_i-1,i}) - (C_{a_j-1,j}-C_{a_j-2,j}) \LT 0\;. 
\qedhere 
\]
\end{proof}

We now have all the pieces needed to prove Theorem~\ref{T-main} by induction.

\begin{proof}[Proof of Theorem~\textup{\ref{T-main}}]
The basis of the induction is established by Claim~\ref{C-basis}, which shows that the theorem holds for $0\le N\le k$. We thus need to show that if the theorem holds for every $(N',k')<(N,k)$, for some $N>k$, then it also holds for $(N,k)$. Note that if $N>k$ and $N\le \binom{n+k}{k}$, then $n\ge 2$. Claim~\ref{C-CNk} shows that for every $\aaa\in Q_{N,k}$ we have $C(\aaa)=(N+1)n - \binom{n+k}{k+1}$. Claim~\ref{C-non-opt} shows that if $\aaa\in P_{N,k}\setminus Q_{N,k}$ then $C(\aaa)>C_{N,k}$. It thus follows that $C_{N,k}=(N+1)n - \binom{n+k}{k+1}$, establishing Theorem~\ref{T-main}(i) and (iii) for $(N,k)$. Claim~\ref{C-diff} then establishes Theorem~\ref{T-main}(ii). This completes the proof of the induction step and hence the proof of Theorem~\ref{T-main}.
\end{proof}

A convenient corollary of Theorem~\ref{T-main} is the following: 

\begin{corollary}\label{C-main}
\textup{(i)} If $N=\binom{n+k}{k}-1$, then $C_{N,k}=\frac{kn}{k+1}(N+1)$.\\
\textup{(ii)} If $\binom{n+k-1}{k}\le N\le \binom{n+k}{k}-1$, then $C_{N,k}\ge \frac{k}{k+1}(n-1)N$ and thus $A_{N,k}\ge \frac{k}{k+1}(n-1)\ge \frac{1}{2}(n-1)$.
\end{corollary}

\begin{proof}
To prove (i) note that if $N=\binom{n+k}{k}-1$ then
\[\textstyle C_{N,k} \EQ \binom{n+k}{k}n-\binom{n+k}{k+1} \EQ \binom{n+k}{k}n-\frac{n}{k+1}\binom{n+k}{k} \EQ \frac{kn}{k+1}\binom{n+k}{k} \EQ \frac{kn}{k+1}(N+1)\;. \]
For each $n\ge 1$ we prove (ii) by induction on~$N$. If $N=\binom{n+k-1}{k}$ then, $C_{N,k}>C_{N-1,k}=\frac{k(n-1)}{k+1}N$. Now, if $\binom{n+k-1}{k}<N<\binom{n+k}{k}$, then by Theorem~\ref{T-main}(ii), 
\[\textstyle C_{N,k} \EQ C_{N-1,k}+n \GE \frac{k(n-1)}{k+1}(N-1)+n \GE \frac{k}{k+1}(n-1)N\;.
\qedhere
\]
\end{proof}

\subsection{Playing optimally - the binomial counter.}\label{sub-play}

An optimal solution of the $(N,k)$-growth game, i.e., a sequence of moves of minimum total cost, can be easily reconstructed with the help of Theorem~\ref{T-main} and Lemma~\ref{L-decompose}. In general such sequences are not unique. The optimal sequence is unique only if $N=\binom{n+k}{k}-1$ for some $n\ge 0$. 

\begin{lemma}
The $(N,k)$-growth game has a unique optimal sequence if and only if $N=\binom{n+k}{k}-1$ for some $n\ge 0$. 
\end{lemma}

\begin{proof}
Recall that $Q_{N,k}$ is the set of optimal final states. We first show that $|Q_{N,k}|=1$ if and only if $N=\binom{n+k}{k}-1$ for some $n\ge 0$. This follows easily as $\sum_{i=1}^k \binom{n-i+1}{i}=\binom{n+k}{k}-1=N$. Thus, if $N=\binom{n+k}{k}-1$, then the only optimal final state is $a_i=\binom{n-i+1}{i}$, for $i\in[k]$. If~$N$ is not of this form then $|Q_{N,k}|>1$ as there is some freedom in choosing the individual $a_i$'s. Finally, if $N=\binom{n+k}{k}-1$, then by Lemma~\ref{L-decompose}, all intermediate states are also of this form and the optimal sequence is unique.
\end{proof}

\newcommand{\Initialize}{\mbox{\it Initialize}}
\newcommand{\Increment}{\mbox{\it Increment}}

The unique optimal sequence of moves for $N=\binom{n+k}{k}-1$ can be expressed using a \emph{binomial counter} whose pseudocode is given in Figure~\ref{P-binomial-counter}. In addition to $(a_1,a_2,\ldots,a_k)$, the sizes of the subarrays, we also keep a counter $(b_1,b_2,\ldots,b_k)$. $\Initialize(k)$ sets $a_1,a_2,\ldots,a_k$ and $b_1,b_2,\ldots,b_k$ to~$0$. It also sets $b_{k+1}$ to $\infty$. $\Increment(k)$ increments the counter and performs a move in the $(N,k)$-growth game as follows: it finds the smallest~$i$ for which $b_i<b_{i+1}$. (It will follow, as we shall see, that $b_1=b_2=\cdots=b_i<b_{i+1}$.) It then performs the $i$-th move in the game, i.e., $a_i\gets 1 + \sum_{j=1}^i a_j$ while $a_1,a_2,\ldots,a_{i-1}\gets 0$. (This corresponds to merging $A_1,A_2,\ldots,A_i$.) It then increments the counter by letting $b_i\gets b_i+1$ and $b_1,b_2,\ldots,b_{i-1}\gets 0$. An example of a binomial counter in action, for $k=4$, is given in Figure~\ref{F-binomial-counter}.

\begin{lemma}\label{L-counter-1}
At any time, $a_i=\binom{b_i+i-1}{i}$ for every $i\in[k]$.
\end{lemma}

\begin{proof}
By induction on the number of steps of the counter. The condition is satisfied initially, since $\binom{i-1}{i}=0$ for $i\in [k]$. Suppose that the claim holds for the current $(a_1,a_2,\ldots,a_k)$ and $(b_1,b_2,\ldots,b_k)$ and that $b_1=b_2=\cdots=b_i=n<b_{i+1}$. Let $a'i=1+\sum_{j=1}^i a_j$ and $b'_i=b_i+1$ be the new values of~$a_i$ and~$b_i$. By the induction hypothesis and Lemma~\ref{L-binom}(ii) we get
\[ \textstyle a'_i \EQ 1+\sum_{j=1}^i a_j \EQ 1 + \sum_{j=1}^i \binom{n+i-1}{i} \EQ \sum_{j=0}^i \binom{n+i-1}{i} \EQ \binom{n+i}{i} \EQ \binom{b'_i+i-1}{i} \;. 
\qedhere 
\]
\end{proof}

We thus get:

\begin{lemma}\label{L-counter-2}
If $b_1=b_2=\cdots b_k=n$ for some $n\ge 0$, then $N=\sum_{i=1}^k a_i = \binom{n+k-1}{k}-1$ and the sequence of moves of the binomial counter is the unique optimal solution of the $(N,k)$-growth game.
\end{lemma}

\begin{figure}[t]
    \centering
    \parbox{2in}{
    \begin{algorithm}[H]
    \Fn{$\Initialize(k)$}
    {
        $a_1,a_2,\ldots,a_k\gets 0$ \;
        $b_1,b_2,\ldots,b_k\gets 0$ \;
        $b_{k+1}\gets \infty$ \;
    }
    \end{algorithm}
    }
    \parbox{3in}{
    \begin{algorithm}[H]
    \Fn{$\Increment(k)$}
    {
        $i\gets \min\{ j\in[k] \mid b_j<b_{j+1} \}$ \;
        \BlankLine
        $a_i\gets 1 + \sum_{j=1}^i a_i$ \;
        $b_i\gets b_i+1$ \;
        \BlankLine
        $a_1,a_2,\ldots,a_{i-1}\gets 0$ \;
        $b_1,b_2,\ldots,b_{i-1}\gets 0$ \;
    }
    \end{algorithm}
    }
    \caption{Pseudocode of the binomial counter.}
    \label{P-binomial-counter}
\end{figure}

\begin{figure}[t]
\centering
\begin{tabular}{ccc}
\includegraphics[scale=0.45,valign=b]{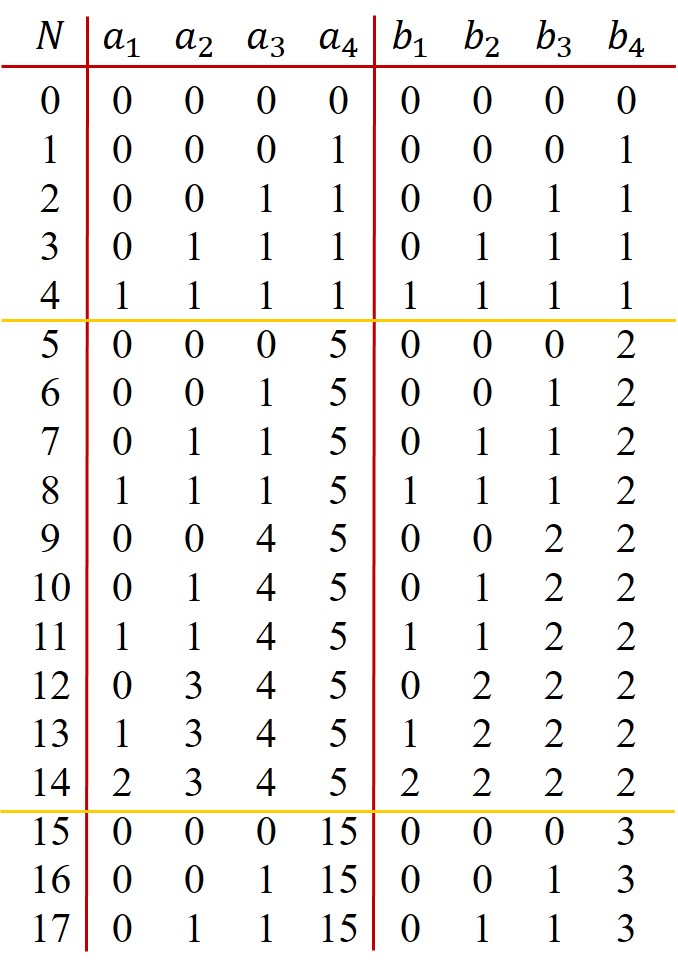} & \hspace*{0cm} &
\includegraphics[scale=0.45,valign=b]{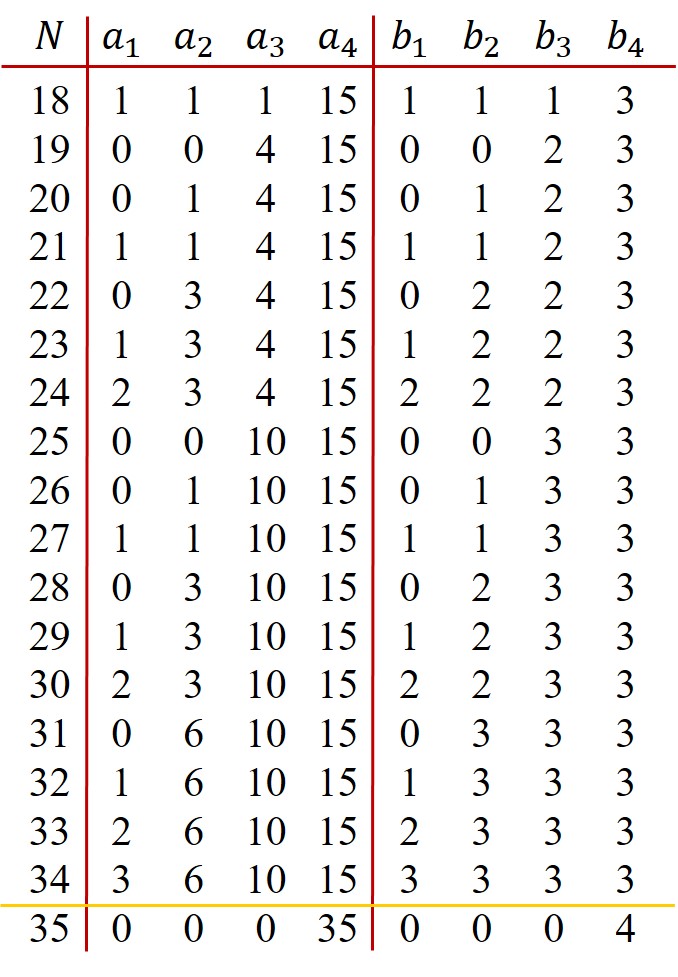}
\end{tabular}
\caption{The binomial counter for $k=4$.}\label{F-binomial-counter}
\end{figure}

\subsection{Extended rules.}\label{sub-extend-rule}

In this section we show that the rules of the $(N,k)$-growth game can be slightly extended without changing the optimal values. 
The third rule of the game, shown in Figure~\ref{F-growth-game}(c) can be written as $(a_1,a_2,\ldots,a_k) \to (0,\ldots,0,1+\sum_{\ell=1}^i a_\ell,a_{j+1},\ldots,a_k)$, for some $i\in[k]$. We first allow a move $(a_1,a_2,\ldots,a_k) \to (0,\ldots,0,a_1,\ldots,a_{i-1},\sum_{\ell=i}^j a_\ell,a_{j+1},\ldots,a_k)$, for $1\le i<j\le k$, in which $A_i,A_{i+1},\ldots,A_j$ are merged into a new $A_j$, the subarrays $A_1,\ldots,A_{i-1}$ become $A_{j-i+1},\ldots,A_{j-1}$ and $A_1,\ldots,A_{j-i}$ become empty. (No new item is added by this rule.) Many of the algorithms described in the paper use such moves. The cost of the move is $\sum_{\ell=i}^j a_\ell$. Note that the order of the items in the subarrays is still maintained.

\begin{lemma}\label{L-extend-1}
Allowing moves $(a_1,a_2,\ldots,a_k) \to (0,\ldots,0,a_1,a_2,\ldots,a_{i-1},\sum_{\ell=i}^j a_\ell,a_{j+1},\ldots,a_k)$, for $1\le i<j\le k$, does not change the optimal values of the $(N,k)$-growth game.
\end{lemma}

\begin{proof} We show that for every $\aaa\in P_{N,k}$ and every $1\le i<j\le k$ we have
\[\textstyle C_k(0,\ldots,0,a_1,a_2,\ldots,a_{i-1},\sum_{\ell=i}^j a_\ell,a_{j+1},\ldots,a_k) \LE  C_k(a_1,a_2,\ldots,a_k) + \sum_{\ell=i}^j a_\ell \;, \]
showing that the new moves do not offer any advantage. By Lemma~\ref{L-decompose} we have $C_k(a_1,a_2,\ldots,a_k)=C_j(a_1,a_2,\ldots,a_j)+C_k(0,\ldots,0,a_{j+1},\ldots,a_k)$ and similarly for the term appearing on the left hand side of the inequality. It is thus enough to prove the inequality for $j=k$. Now
\begin{align*}
    & \textstyle C_k(0,\ldots,0,a_1,a_2,\ldots,a_{i-1},\sum_{\ell=i}^k a_\ell) \\
    \LE\; & \textstyle C_{k-1}(0,\ldots,0,a_1,a_2,\ldots,a_{i-1}) + C_k(0,\ldots,0,\sum_{\ell=i}^k a_\ell) \\
    \LE\; & \textstyle C_{k-1}(0,\ldots,0,a_1,a_2,\ldots,a_{i-1}) + C_k(0,\ldots,0,a_i,a_{i+1},\ldots,a_k-1) + \sum_{\ell=i}^k a_\ell \\
    \LE\; & \textstyle C_{k-1}(0,\ldots,0,a_1,a_2,\ldots,a_{i-1}) + C_k(0,\ldots,0,a_i,a_{i+1},\ldots,a_k) + \sum_{\ell=i}^k a_\ell \\
    \LE\; & \textstyle C_k(a_1,a_2,\ldots,a_k) + \sum_{\ell=i}^j a_\ell \;.
\end{align*}
The last inequality follows by expanding the terms using Lemma~\ref{L-decompose}(iv). Note that $C_k(a_1,a_2,\ldots,a_k)$ expands to $\sum_{i:a_i>0} C_{a_i-1},i + \sum_{i=1}^k a_i$, while the previous line expands to $\sum_{\ell:a_\ell>0} C_{a_\ell-1,k_\ell} + \sum_{\ell=1}^k a_\ell$, where $k_\ell\ge \ell$, for every $\ell\in [k]$. More precisely, if $1\le \ell\le i-1$, then $k_\ell=(k-i)+\ell\ge\ell$, and if $i\le \ell\le k$, then $k_\ell=\ell$. (Clearly $C_{N,k'}\le C_{N,k}$, for $k'\ge k$.)
\end{proof}

We next show that allowing general merge operations that do not necessarily respect the order of the items in the subarrays also does change the solution of the $(N,k)$-growth game. More precisely, for every $I=\{i_1,i_2,\ldots,i_r\}\subseteq [k]$, where $i_1<i_2<\cdots<i_r$, we allow an $I$-move that merges $A_{i_1},A_{i_2},\ldots,A_{i_r}$, and a new item. The new subarray becomes the new $A_{i_r}$. The subarrays $A_{i_1},A_{i_2},\ldots,A_{i_{r-1}}$ become empty. All other subarrays are unchanged. In other words, $a_{i_r}\gets 1+\sum_{\ell=1}^r a_{i_\ell}$, while $a_{i_\ell}\gets 0$, for $\ell\in[r-1]$. The standard moves in the growth games are $[i]$-moves, for $i\in[k]$. (Note that the decision to let $A_{i_r}$ be the new merged array is arbitrary, since $i_r$ is just the index assigned to the new array. This, however, is the convenient choice, since otherwise we need to define a new cost function $\bar{C}(\aaa)$ on states equal to the minimum of $C(\aaa')$, where~$\aaa'$ is a permutation of~$\aaa$.) We can also allow a variant of an $I$-move in which a new item is not added. The proof of the following lemma is essentially the same as the proof of Lemma~\ref{L-extend-1} and is therefore omitted.

\begin{lemma}\label{L-extend-2}
Allowing $I$-moves for any $I\subseteq [k]$, does not change the optimal values of the $(N,k)$-growth game.
\end{lemma}

With these extensions, the growth game now captures all the moves that can be performed by standard implementations. (See Definition~\ref{D-standard}.)

\section{Lower bound on the amortized cost of grow operations}\label{S-lower-grow}

We can now prove an $\Omega(r)$ lower bound on the amortized cost of grow operations for standard implementations that use only $N+O(rN^{1/r})$ space to store an array of size~$N$, showing that the constructions of Section~\ref{S-general-r} are essentially optimal.

\begin{theorem}
Any standard resizable array data structure that uses only $N+O(rN^{1/r})$ space to store an array of size~$N$, where $r=O(\log N)$, must have an amortized cost of $\Omega(r)$ for grow operations.
\end{theorem}

\begin{proof}
Assume that the extra space used by the data structure is $rN^{1/r}$. (If the extra space is $\alpha r N^{1/r}$ for some constant $\alpha>1$, apply the theorem for $r-1$, relying on the inequality $\alpha r N^{1/r}<(r-1)N^{1/(r-1)}$ for sufficiently large~$N$.) 
Since the operations of any standard implementation correspond to moves in the (extended) $(N,rN^{1/r},rN^{1/r})$-growth game, Corollary~\ref{C-main} implies that the amortized cost of grow operations is at least $A_{N,rN^{1/r},rN^{1/r}}=A_{\frac{1}{r}N^{1-1/r},rN^{1/r}}\ge \frac{1}{2}n$, provided that $\binom{n+rN^{1/r}}{n}< \frac{1}{r}N^{1-1/r}$. We thus want to show that when $n=\beta r$, for some $\beta>0$, we have $\binom{\beta r+ rN^{1/r}}{\beta r} < \frac{1}{r}N^{1-1/r}$ for sufficiently large~$N$.

To upper bound binomial coefficients, we use the well known inequality $\binom{n}{k}\le \bigl(\frac{\e n}{k}\bigr)^k$. We now have:
\[ \binom{\beta r+ rN^{1/r}}{\beta r} \LE \binom{(\beta+1)rN^{1/r}}{\beta r} \LE \left(\frac{\e(\beta+1)N^{1/r}}{\beta}\right)^{\beta r}
\EQ \left(\frac{\e(\beta+1)}{\beta}\right)^{\beta r} N^\beta\;.\] 
Choosing for example $\beta=\frac12$, and assuming that $3\le r\le \alpha\log N$, for an appropriate choice of~$\alpha$, we get that $r(3\e)^{r/2}N^{1/2}\le N^{3/4}\le N^{1-1/r}$, as required.
\end{proof}

\section{Concluding remarks}\label{S-concl}

We have presented a family of resizable array implementations that provide optimal trade-offs between the space needed to store an array, the space needed to resize, i.e., grow or shrink, the array, and the amortized cost of grow and shrink operations, while still maintaining $O(1)$ worst-case access cost. We believe this solves a very fundamental problem almost completely.

Our lower bound on the amortized cost of grow operations applies only to what we call standard algorithms. We believe that these lower bounds can be extended to cover all algorithms, assuming a suitable \emph{incompressability} assumption on items and pointers. (See the discussion after the proof of Theorem~\ref{T-lower1}.) In other words, we believe that non-standard algorithms do not offer any advantage over standard algorithms and that to prove this rigorously should not be conceptually hard.

To obtain the amortized lower bounds on grow operations, we defined an interesting growth game and then solved it exactly. While we believe that the exact solution of the game is illuminating, it would be interesting to know if there is a simpler way of obtaining, perhaps by induction, lower bounds on the values of the game that would be sufficient to prove asymptotically optimal amortized lower bounds.

\bibliographystyle{plain}
\bibliography{bibliography}

\end{document}